\documentclass[11pt]{article}

\usepackage{lineno}
\usepackage{amssymb}
\usepackage{amsmath}
\usepackage{amsthm}
\usepackage{epsfig}
\usepackage{graphicx}
\usepackage{graphics}
\usepackage{float}
\usepackage{subfigure}
\usepackage{multirow}
\usepackage{color}
\usepackage{lineno}
\usepackage{fullpage}
\usepackage{makeidx}
\usepackage{xspace}
\usepackage{wrapfig}
\usepackage{blkarray, bigstrut}
\usepackage{pdfpages}
\usepackage{imakeidx}
\usepackage{afterpage}
\usepackage{url}
\usepackage{nomencl}
\makenomenclature
\usepackage[normalem]{ulem} 
\usepackage[title]{appendix}
\usepackage[pdftex,colorlinks=true, urlcolor=black,linkcolor=black,citecolor=blue,pdfstartview=FitH,plainpages=true]{hyperref} 
\usepackage{bbm}
\usepackage{etoolbox}
\renewcommand\nomgroup[1]{%
  \item[\bfseries
  \ifstrequal{#1}{M}{Matrices and Vectors}{%
  \ifstrequal{#1}{T}{Tensors}{%
  \ifstrequal{#1}{R}{Real Numbers}{}}}%
]}

\usepackage{xcolor}
\usepackage{smartdiagram}
\usepackage{verbatim}
\usepackage{tikz}
\usepackage{array}
\usepackage{booktabs}
\usepackage[round]{natbib}
\usepackage{authblk}

\def\vec{\hbox{\rm vec}}
\def\Cov{{\rm Cov}}
\def\Var{{\rm Var}}
\def\col{\hbox{\rm Col}}

\DeclareMathOperator{\rank}{rank}
\DeclareMathOperator{\tr}{tr}

\newcommand{\as}{\ensuremath{\,\text{a.s.}}}

\newcolumntype{P}[1]{>{\centering\arraybackslash}p{#1}}
\DeclareMathOperator*{\argmin}{argmin}
\newtheorem{definition}{Definition}
\newtheorem{assumption}{Assumption}

\newtheorem{lemma}{Lemma}
\newtheorem{corollary}{Corollary}
\newtheorem{theorem}{Theorem}

\newcommand{\h}[1]{\boldsymbol{#1}}

\newcommand{\hA}{\h{A}}

\newcommand{\hX}{\h{X}}

\newcommand{\bA}{\h{A}}
\newcommand{\bB}{\h{B}}
\newcommand{\bC}{\h{C}}
\newcommand{\bD}{\h{D}}
\newcommand{\bE}{\h{E}}

\newcommand{\bG}{\h{G}}
\newcommand{\bH}{\h{H}}
\newcommand{\bI}{\h{I}}

\newcommand{\bL}{\h{L}}
\newcommand{\bM}{\h{M}}

\newcommand{\bP}{\h{P}}
\newcommand{\bQ}{\h{Q}}
\newcommand{\bR}{\h{R}}
\newcommand{\bS}{\h{S}}

\newcommand{\bU}{\h{U}}
\newcommand{\bV}{\h{V}}
\newcommand{\bW}{\h{W}}
\newcommand{\bX}{\h{X}}
\newcommand{\bY}{\h{Y}}
\newcommand{\bZ}{\h{Z}}

\newcommand{\bc}{\h{c}}

\newcommand{\bw}{\h{w}}
\newcommand{\bx}{\h{x}}
\newcommand{\by}{\h{y}}
\newcommand{\bz}{\h{z}}
\newcommand{\bd}{\h{d}}
\newcommand{\bepsilon}{\h{\epsilon}}
\newcommand{\bbeta}{\h{\beta}}
\newcommand{\balpha}{\h{\alpha}}
\newcommand{\bzero}{\h{0}}

\newcommand{\cX}{{\mathcal X}}
\newcommand{\cY}{{\mathcal Y}}

\title{\textbf{Cointegrated Matrix Autoregression Models}}
\author{Zebang Li and Han Xiao}
\affil{Rutgers University}
\date{}

\begin{document}

\maketitle

\begin{abstract}
 We propose a novel cointegrated autoregressive model for matrix-valued time series, with bi-linear cointegrating vectors corresponding to the rows and columns of the matrix data. Compared to the traditional cointegration analysis, our proposed matrix cointegration model better preserves the inherent structure of the data and enables corresponding interpretations. To estimate the cointegrating vectors as well as other coefficients, we introduce two types of estimators based on least squares and maximum likelihood. We investigate the asymptotic properties of the cointegrated matrix autoregressive model under the existence of trend and establish the asymptotic distributions for the cointegrating vectors, as well as other model parameters. We conduct extensive simulations to demonstrate its superior performance over traditional methods. In addition, we apply our proposed model to Fama-French portfolios and develop a effective pairs trading strategy. 

\end{abstract}
\newpage

\section{Introduction}\label{SEC:introduction}
Cointegration refers to a long-run, stable relationship between two or more non-stationary time series data. It has become more and more observed and recognized in econometrics, finance, and other fields. The term `cointegration' was first introduced in the late 1980s, see \cite{granger1981some}, \cite{granger1983time} and \cite{engle1987co}. Since then, the phenomenon of cointegration has been widely studied and various methods have been developed for testing and estimation. Statistical analyses of cointegration models were developed by \cite{johansen1988statistical}, \cite{johansen1991estimation}, and \cite{johansen1995likelihood}, among others. They derived the testing of cointegration ranks between multiple time series variables and the maximum likelihood estimator of the cointegrating vectors through a vector error correction model (VECM), which is now widely used and considered to be standard in the literature.

Recently, \cite{wang2014martingale} studied the martingale limit theorem for a nonlinear cointegrated regression model. \cite{doornik2017maximum} proposed a novel estimation procedure for the I(2) cointegrated vector autoregressive model. \cite{franchi2017improved} studied the vector cointegration model allowing for multiple near unit roots using adjusted quantiles. \cite{cai2017new} studied the bivariate threshold VAR cointegration model. \cite{zhang2019identifying} proposed a new method for identifying cointegrated components of nonstationary time series. \cite{she2020inference} studied the full rank and reduced rank least squares estimators of the heavy-tailed vector error correction models. \cite{gao2021modeling} proposed a new procedure to build factor models for high-dimensional unit-root time series. \cite{guo2022automated} proposed an automated estimation method of heavy-tailed vector error correction models. To the best of our knowledge, there has been no prior research on cointegration for matrix-valued time series. This paper aims to fill this gap in the literature by presenting a cointegrated matrix autoregressive models (CMAR).

In this paper, we present a novel extension of the cointegrated VAR model, which focuses on the cointegration of matrix-valued time series, denoted by $\bX_t$, a $d_1 \times d_2$ matrix observed at time $t$. We introduce the error correction model in bilinear form, with row-wise and column-wise cointegrating vectors $\bbeta_1$ and $\bbeta_2$, respectively, which enable the transformation of $\bX_t$ into a stationary process $\bbeta_1'\bX_t\bbeta_2$, while the original process is a non-stationary I(1) process. 
Our proposed approach presents a promising alternative to traditional cointegration analysis, which is in line with the recent work on autoregressive models for stationary processes by \cite{chen2021autoregressive}, \cite{MARRR}. 

To estimate the coefficient matrices, especially the cointegrating vectors, for the proposed model, we introduce two different estimators based on the least squares and likelihood principles respectively. We obtain both the least squares estimator (LSE) and the maximum likelihood estimator (MLE) using iterative algorithms. We establish the asymptotics for all the proposed estimators, with a special focus on the estimators of the cointegration space, represented by $\hat\bbeta (\hat\bbeta' \hat\bbeta)^{-1} \hat\bbeta'$. Our analysis reveals that the process $\bX_t$ may exhibit different properties in different directions. These findings have important implications for the interpretation and analysis of the cointegration phenomenon of the matrix-valued time series. The proposed estimators are evaluated through extensive simulations, with better performance over traditional vector cointegration methods.

To demonstrate the practical application of our proposed approach, we develop a pairs trading strategy using real-world data from the Fama-French portfolios. Given that the stocks can be naturally grouped into a matrix format based on factors such as Operating Profitability and Book-to-Market ratio, we apply our model and build an equilibrium relationship among the portfolios. Leveraging this stationary relationship, we can execute short and long trades accordingly. The back-testing results demonstrate that our strategy has the potential to outperform the market, especially during bear markets.

The rest of this paper is organized as follows. We briefly review the cointegrated vector autoregressive models, and then introduce the cointegrated matrix autoregressive model in Section~\ref{SEC:CMARmodels}, and discuss its basic properties.
The estimators are presented in Section~\ref{SEC:estimation}. Asymptotic properties of the estimators are established in Section~\ref{SEC:Theorems}. In Section~\ref{SEC:simulation}, we carry out extensive numerical studies to demonstrate the properties and performance of the model and the corresponding estimators, and compare with the classical vector cointegration models. In Section~\ref{SEC:trading}, we develop a pairs trading strategy based on the proposed matrix cointegration model. All the proofs and some additional theorems and simulations are collected in Appendix~\ref{Appendix}.

\textbf{Notations}. For ease of reading, here we highlight some notations that are frequently used. Throughout this paper, we use the bold uppercase letters for matrices (e.g. $\bX_t$), and the bold lowercase letters for vectors (e.g. $\bx$). We use $\|\cdot\|_F$ and $\|\cdot\|_s$ to denote the matrix Frobenius norm and spectral norm respectively. The Kronecker product of two matrices is denoted by $\otimes$. Note that the notations $\bG_{ij}$, $\bB_{ij}$ refer to some new matrices to be defined later, and they should not be regarded as entries of $\bG$ or $\bB$. For any $d_1 \times d_2$ matrix $\bbeta$ of full rank (assuming $d_1 \leq d_2$), we denote by $\bbeta_{\bot}$ a $d_1 \times (d_1 - d_2)$ matrix of rank $d_1 - d_2$ such that $\bbeta' \bbeta_{\bot} = \bzero$. We define $\bar{\bbeta} = \bbeta(\bbeta' \bbeta)^{-1}$, such that $\bbeta'\bar{\bbeta} = \bI_{d_2}$, and $\mathbb{P}_{\beta} = \bar{\bbeta} \bbeta'$ is the projection of $\mathbb{R}^{d_1}$ onto the space spanned by the columns of $\bbeta$. We use $\Delta$ to denote the difference operator such that $\Delta \bx_t = \bx_t - \bx_{t-1}$. 

\section{Cointegrated Matrix Autoregression Models}\label{SEC:CMARmodels}
\subsection{Cointegrated Vector Autoregression Models}
To fix notations, we briefly introduce some basic concepts about the cointegrated vector autoregression models (CVAR). For a more thorough analysis on the vector cointegration models, see \cite{johansen1995likelihood}.

Consider a vector time series $\{\bx_t\}$. If we allow unit roots in the characteristic polynomial of the process that $\bx_t$ is non-stationary, but some linear combination $\bbeta' \bx_t$, $\bbeta \neq \bzero$, can be made stationary by a suitable choice of $\bx_0$, then $\bx_t$ is called cointegrated and $\bbeta$ is the cointegrating vectors. The number of linearly independent cointegrating vectors is called the cointegration rank, and the space spanned by the cointegrating vectors is the cointegration space. We define the order of cointegration as follows.
\begin{definition}
The vector process $\{\bx_t \in \mathbb{R}^{d}\}$ is called integrated of order zero I(0) if $\bx_t = \sum_{i=0}^{\infty} \bC_i \bepsilon_{t-i}$ such that $\sum_{i=0}^{\infty} \bC_i \neq 0$. It is called integrated of order one I(1) if $\Delta \bx_t$ is I(0).
\end{definition}
The autoregressive equations are given in the error-correction form
\begin{equation}\label{CVAR}
    \Delta\bx_t = \Pi \bx_{t-1} + \sum_{i=1}^{k} \Gamma_i \Delta\bx_{t-i} +  \bd + \bepsilon_t, t=1.\ldots,T,
\end{equation}
where the $\bepsilon_t$ are IID $N_p(0,\Sigma)$, $\bd$ is a constant term. The equations determine the process $\bx_t$ as a function of initial values $\bx_0,\ldots,\bx_k$, and the $\bepsilon_t$. Let $A(z)$ denote the vector polynomial derived from (\ref{CVAR}),
\begin{equation*}\label{Az}
    A(z) = (1-z) \bI_d - \Pi z - \sum_{i=1}^{k} \Gamma_i z^{i} (1-z).
\end{equation*}
The basic assumption which will be used throughout is as follows,
\begin{assumption}
The characteristic polynomial satisfies the condition that if $|A(z)|=0$, then either $z>1$ or $z=1$.
\end{assumption}
If $z=1$ is a root we say that the process has a unit root. The following representation theorem in \cite{johansen1995likelihood} provides sufficient and necessary conditions on the coefficients of the autoregressive model for the process to be integrated of order one.
\begin{theorem}\label{Theorem:represent}
If $|A(z)|=0$ implies that $|z|>1$ or $z=1$, and $\rank \Pi = r < d$, such that $\Pi=\balpha \bbeta'$, where $\balpha$, $\bbeta$ are $d \times r$ matrices of rank $r$. A necessary and sufficient condition that $\Delta\bx_t - \mathbb{E}(\Delta\bx_t)$ and $\bbeta'\bx_t - \mathbb{E}(\bbeta'\bx_t)$ can be given initial distributions such that they become I(0) is that $\balpha_{\bot}' \Gamma \bbeta_{\bot}$ has full rank. In this case the solution of (\ref{CVAR}) has the representation
\begin{equation}
    \bx_t = \bC \sum_{i=1}^{t} (\bepsilon_i + \bd) + C(L) (\bepsilon_t + \bd) + \mathbb{P}_{\beta_{\bot}} \bx_0,
\end{equation}
where $\bC = \bbeta_{\bot} (\balpha_{\bot}'\Gamma\bbeta_{\bot})^{-1}\balpha_{\bot}'$, $\Gamma = \bI - \sum_{i=1}^{k}\Gamma_i$. Thus $\bx_t$ is a cointegrated I(1) process with cointegrating vectors $\bbeta$.
\end{theorem}
By building upon the aforementioned results, a comprehensive framework was established for formulating and estimation in the context of cointegration, as detailed in works such as \citep{johansen1988statistical, johansen1991estimation, johansen1995likelihood, johansen1999testing}, among others. In the following sections, we extend these works to the case of matrix time series.

\subsection{Cointegrated MAR Models}
At each time $t$, a matrix $\bX_t \in \mathbb{R}^{d_1 \times d_2}$ is observed. We introduce the cointegrated matrix autoregression models (CMAR) of the form
\begin{equation}\label{model form 2}
    \Delta \bX_t = \bA_1 \bX_{t-1} \bA_2' + \sum_{i=1}^{k} \bB_{i1} \Delta \bX_{t-i} \bB_{i2}' + \bD + \bE_t, \ t=1,\ldots,T,
\end{equation}
where $\bA_j$ are $d_j \times d_j$ coefficient matrices of ranks $r_j \le d_j$, and $\bB_{ij}$ are $d_j \times d_j$ coefficient matrices without rank constraints, for $j=1,2$, $i=1,\ldots,k$. The $d_1 \times d_2$ matrix $\bD$ is a constant term. $\bE_t$ is a $d_1 \times d_2$ matrix white noise. The low rank assumption implies that $\bA_j = \balpha_j \bbeta_j'$ where $\balpha_j$ and $\bbeta_j$ are both $d_j \times r_j$ full rank matrices. $\bbeta_1$ and $\bbeta_2$ are cointegrating vectors.

The model (\ref{model form 2}) offers a parsimonious representation of the cointegrated vector autoregressive models (CVAR). Taking vectorization on both sides of (\ref{model form 2}), it becomes
\begin{equation}\label{vecCMAR}
    \Delta \vec(\bX_t) = (\bA_2 \otimes \bA_1) \vec(\bX_{t-1}) + \sum_{i=1}^{k}(\bB_{i2} \otimes \bB_{i1}) \Delta \vec(\bX_{t-i}) + \vec(\bD) + \vec(\bE_t).
\end{equation}
In comparison to the CVAR model (\ref{CVAR}), the CMAR model imposes further constraints on the coefficients by requiring them to take the form of a Kronecker product, such that $\Pi = \bA_2 \otimes \bA_1$, $\Gamma_i = \bB_{i2} \otimes \bB_{i1}$, and the cointegrating vectors $\bbeta = \bbeta_2 \otimes \bbeta_1$. If the vector process $\vec(\bX_t)$ is integrated of order one such that $\bbeta' \vec(\bX_t)$ is stationary, then it is equivalent to say the matrix process $\bbeta_1' \bX_t \bbeta_2$ is stationary. The concept of I(0) and I(1) for matrix-valued time series can be naturally extended from the vector case.
\begin{definition}
The $d_1 \times d_2$ matrix process $\{\bX_t\}$ is called integrated of order zero I(0) if $\vec(\bX_t) = \sum_{i=0}^{\infty} \bC_i \bepsilon_{t-i}$ such that $\sum_{i=0}^{\infty} \bC_i \neq 0$. It is called integrated of order one I(1) if $\Delta \bX_t$ is I(0).
\end{definition}
The characteristic polynomial for the process $\bX_t$ given in (\ref{model form 2}) is
\begin{equation*}
   A(z) = (1-z) \bI - (\bA_2 \otimes \bA_1)z - \sum_{i=1}^{k} (\bB_{i2} \otimes \bB_{i1}) (1-z)z^i.
\end{equation*}
The unit root condition of a matrix time series can be expressed as follows: 
\begin{assumption}\label{rootA}
If $\left|A(z)\right| = 0$, then either $|z| > 1$ or $z = 1$.
\end{assumption}
A necessary and sufficient condition that $\Delta\bX_t - \mathbb{E}(\Delta \bX_t)$ and $\bbeta_1' \bX_t \bbeta_2 - \mathbb{E}(\bbeta_1' \bX_t \bbeta_2)$ can be given initial distributions such that they become I(0) is that $\vec(\bX_t)$ satisfies the conditions of Theorem~\ref{Theorem:represent}. It is worth noting that if $\bX_t$ is I(0), then $\bbeta_1' \bX_t \bbeta_2$ is stationary. However, either $\bbeta_1' \bX_t$ or $\bX_t \bbeta_2$ may not be stationary. Therefore, both row-wise and column-wise cointegrating vectors are necessary.

\subsection{Identifiability}
If we require that $\|\bA_1\|_F=1$ and $\|\bB_{j1}\|_F=1$, $1\leq j\leq k$, then each coefficient matrix $\bA_1$, $\bA_2$, $\bB_{j1}$, $\bB_{j2}$ is identified up to a sign change. We also assume $\bbeta_i$ is the left singular vectors of $\bA_i$ such that $\bbeta_i' \bbeta_i = \bI_{r_i}$, $i=1,2$. The estimators $\hat\bbeta_i$ can be normalized by any non-zero matrix $\h{c}$, denoted as $\hat\bbeta_{ic} = \hat\bbeta_{i} (\h{c}' \hat\bbeta_{i})^{-1}$.
Note that the projection matrix $\mathbb{P}_{\beta_i} = \bbeta_i (\bbeta_i' \bbeta_i)^{-1} \bbeta'_i$ would remain unique regardless of how we normalize $\bbeta_i$, i.e., the choice of non-zero matrix $\bc$ would not affect the projection matrix. In section~\ref{SEC:Theorems}, we first derive asymptotic results for $\tilde\bbeta_{i} = \hat\bbeta_i (\bbeta_i' \hat\bbeta_i)^{-1}$, such that $\tilde\bbeta_i$ is normalized by $\bbeta_{i}' \tilde\bbeta_i = \bI_{r_i}$. We then extend this to the estimators $\hat\bbeta_{ic} = \hat\bbeta_i (\bc' \hat\bbeta_i)^{-1}$ with any normalizing matrix $\bc$ that $\bc' \hat\bbeta_{ic} = \bI_{r_i}$. Then we establish the asymptotic distribution for the projection matrix $\hat{\mathbb{P}}_{\beta_i} = \hat\bbeta_{i} (\hat\bbeta_i' \hat\bbeta_i)^{-1} \hat\bbeta_i'$.

\section{Estimation}\label{SEC:estimation}
We propose to use the alternating algorithm, updating one, while holding the other fixed. Specifically, suppose in the model (\ref{model form 2}), $\bA_2$, $\bB_{i2}$ are given, we discuss how to estimate $\bA_1$, $\bB_{i1}$, and $\bD$. The previous studies by \cite{chen2020autoregressive}, \cite{Li2021}, and \cite{xiao2021} have also employed a similar alternating algorithm to estimate the coefficients, but their analysis assumed that the process ${\bX_t}$ was stationary, while our model is in the error-correction form. To avoid confusion, we adopt the convention that $\bA[,j]$ denotes the $j$-th column of $\bA$. The $j$-th column of the model equation (\ref{model form 2}) can be written as
\[\Delta\bX_{t}[,j] = \bA_1 \bX_{t-1} \bA_2'[,j] + \sum_{i=1}^{k} \bB_{i1} \Delta\bX_{t-1} \bB_{i2}'[,j] + \bD[,j] + \bE_t[,j],\]
which can be viewed as the cointegrated vector autoregressive models when $\bA_2$, $\bB_{i2}$ are fixed. For the estimation, in Section~\ref{SEC:mle}, we introduce the MLE estimator when $\Cov(\vec(\bE_t))$ is separable, which is related to maximum likelihood analysis in \cite{johansen1995likelihood} and canonical correlation analysis in \cite{anderson2003introduction}, \cite{velu2013multivariate}. Then we consider the least squares method in Section~\ref{sec:lse}. The MLE and LSE estimators are respectively denoted as $\hat\bA_{i}^{\text{ml}}$ and $\hat\bA_{i}^{\text{ls}}$.

\subsection{Maximum Likelihood Estimation}\label{SEC:mle}
We commence our analysis of the likelihood by examining a basic model,
\begin{equation}\label{model form 1}
    \Delta \bX_t = \bA_1 \bX_{t-1} \bA_2' + \bE_t, \ t=1,\ldots,T,
\end{equation}
which we later extend it to the model (\ref{model form 2}). We assume that the covariance matrix $\Sigma_e$ of $\vec{(\bE_t)}$ takes the form of a Kronecker product
\begin{equation}\label{sigma_e}
    \Sigma_e = \Sigma_2 \otimes \Sigma_1.
\end{equation}
Under the normality, the log likelihood of the model (\ref{model form 1}) is given apart from a constant by
\begin{equation}\label{like}
    -d_2(T-1) \log|\Sigma_1| - d_1(T-1)\log|\Sigma_2| - \sum_{t=2}^{T}\tr{\left[\Sigma_{1}^{-1} \left(\Delta \bX_t -\bA_1 \bX_t \bA_2'\right)\Sigma_2^{-1}\left(\Delta \bX_t -\bA_1 \bX_t \bA_2'\right)'\right]}.
\end{equation}
We describe an iterative procedure to estimate $\bA_1$ and $\Sigma_1$ given $\bA_2$ and $\Sigma_2$. We rewrite the model by columns as
$$\left( \Delta\bX_t \Sigma_1^{-1/2} \right) [,j] = \bA_1 \left(\bX_{t-1} \bA_2' \Sigma_1^{-1/2} \right) [,j] + \left( \bE_t \Sigma_1^{-1/2}\right)[,j].$$
Let $\by_{tj} = \left(\Delta\bX_t \Sigma_1^{-1/2} \right)[,j]$, $\bx_{tj} = \left(\bX_{t-1} \bA_2' \Sigma_1^{-1/2} \right)[,j]$ and $\bepsilon_{tj} = \left(\bE_t \Sigma_1^{-1/2} \right) [,j]$, then the preceding equation can be viewed as a cointegrated vector autoregressive model with IID errors:
\begin{equation}
    \by_{tj} = \bA_1 \bx_{tj} + \bepsilon_{tj}, \quad 2 \le t \le T, 1\le j \le d_2.
\end{equation}
Then we can apply the estimation procedure of Maximum Likelihood Estimation (MLE) in the classical vector cointegration analysis \citep{johansen1995likelihood}, which is also similar to the reduced rank regression \citep{anderson2003introduction, velu2013multivariate}. Let
\begin{align*}
    \tilde{\bS}_{xx} &= \sum_t \sum_j \bx_{tj} \bx_{tj}' = \sum_t \bX_{t-1} \bA_2' \Sigma_2^{-1} \bA_2 \bX_{t-1}', \\
    \tilde{\bS}_{yx} &= \sum_t \sum_j \by_{tj} \bx_{tj}' = \sum_t \Delta \bX_{t} \Sigma_2^{-1} \bA_2 \bX_{t-1}', \\
    \tilde{\Sigma}_{\epsilon \epsilon} &= \sum_t \left(\Delta \bX_t  - \tilde{\bA}_1 \bX_{t-1} \bA_2'\right) \Sigma_2^{-1}\left(\Delta \bX_t  - \tilde{\bA}_1 \bX_{t-1} \bA_2'\right)',
\end{align*}
where the full rank estimation $\tilde\bA_1 = \tilde\bS_{yx} \tilde\bS_{xx}^{-1}$. Take $\tilde{\bU}:=[\tilde{U}_1, \tilde{U}_1,\ldots,\tilde{U}_{k_1}]$, where $\tilde{U}_j$ is the $j$-th leading unit eigenvector of $\tilde{\Sigma}_{\epsilon \epsilon}^{-1/2} \tilde{\bS}_{yx} \tilde{\bS}_{xx}^{-1} \tilde{\bS}_{xy} \tilde{\Sigma}_{\epsilon \epsilon}^{-1/2}$. Then $\bA_1$, $\Sigma_1$ can be updated as
\begin{align*}
    \hat{\bA}_{1}^{\text{ml}} &= \tilde{\Sigma}_{\epsilon}^{1/2} \tilde\bU \tilde\bU' \tilde{\Sigma}_{\epsilon}^{-1/2} \tilde\bS_{yx} \tilde\bS_{xx}^{-1},\\
    \hat{\Sigma}_1 &= \frac{1}{T-1} \sum_t \left(\Delta\bX_t - \hat{\bA}_1 \bX_{t-1} \bA_2' \right)\Sigma_2^{-1}\left(\Delta\bX_t - \hat{\bA}_1 \bX_{t-1} \bA_2'  \right)'. 
\end{align*}
Given $\bA_1$ and $\Sigma_1$, an update of $\bA_2$ and $\Sigma_2$ can be similarly obtained.

Next, we extend the aforementioned analysis to model (\ref{model form 2}), which involves additional differentiation of previous observations. Similarly, we can express model (\ref{model form 2}) by columns as follows:
\begin{align*}
    \left( \Delta\bX_t \Sigma_2^{-1/2} \right)[,j] = \bA_1 \left(\bX_{t-1} \bA_2' \Sigma_2^{-1/2} \right) [,j] &+ \sum_{i=1}^{k} \bB_{i1} \left(\Delta\bX_{t-i} \bB_{i2}' \Sigma_2^{-1/2} \right) [,j] \\
    &+  \bD  \Sigma_2^{-1/2} [,j] + \left( \bE_t \Sigma_2^{-1/2}\right)[,j].
\end{align*}
Let $\by_{tj} = \left(\Delta\bX_t \Sigma_2^{-1/2} \right)[,j]$, $\bx_{tj} = \left(\bX_{t-1} \bA_2' \Sigma_2^{-1/2} \right)[,j]$, and $\bz_{tj}'$ be stacked vector variables 
$$\left(\Delta\bX_{t-1}\bB_{12}' \Sigma_2^{-1/2} \right)[,j]',\ldots,\left(\Delta\bX_{t-k}\bB_{k2}' \Sigma_2^{-1/2} \right)[,j]', \Sigma_2^{-1/2}[,j]',$$
with length $d_1k+d_2$.
Let $\Psi_1$ be the coefficient corresponding to $\bz_{tj}$ with dimension $d_1 \times (d_1k+d_2)$, which is the matrix consisting of $\bB_{11},\ldots,\bB_{k1},\bD$. Let $\bepsilon_{tj} = \left(\bE_t \Sigma_1^{-1/2} \right) [,j]$. 
Thus, the above equation can be regarded as a cointegrated vector autoregressive model in the error-correction form,
\begin{equation}
    \by_{tj} = \bA_1 \bx_{tj} + \Psi_1 \bz_{tj} + \bepsilon_{tj}, \quad 2 \le t \le T, 1\le j \le d_2.
\end{equation}
Apart from the previously defined $\tilde{\bS}_{xx}$ and $\tilde{\bS}_{yx}$, we require the following notations:
\begin{align*}
    \tilde{\bS}_{yz} &= \sum_t \begin{pmatrix}
    \Delta\bX_{t} \Sigma_2^{-1} \bB_{12} \Delta \bX_{t-1}' & \cdots & 
    \Delta\bX_{t} \Sigma_2^{-1} \bB_{k2} \Delta \bX_{t-k}'
    & \Delta\bX_{t} \Sigma_2^{-1}
    \end{pmatrix}, \\
    \tilde{\bS}_{xz} &= \sum_t \begin{pmatrix}
    \bX_{t-1} \bA_{2}' \Sigma_2^{-1} \bB_{12} \Delta \bX_{t-1}' & \cdots & \bX_{t-1} \bA_{2}' \Sigma_2^{-1} \bB_{k2} \Delta \bX_{t-k}' & \bX_{t-1} \bA_{2}' \Sigma_2^{-1} 
    \end{pmatrix}, \\
    \tilde{\bS}_{zz} &= 
    \sum_t \begin{pmatrix}
    \Delta\bX_{t-1} \bB_{12}' \Sigma_2^{-1} \bB_{12} \Delta \bX_{t-1}' &\cdots& \Delta\bX_{t-1} \bB_{12}' \Sigma_2^{-1} \bB_{k2} \Delta \bX_{t-k}' &\Delta\bX_{t-1} \bB_{12}' \Sigma_2^{-1}\\
    \cdots & \cdots & \cdots & \cdots \\
    \Delta\bX_{t-k} \bB_{k2} \Sigma_2^{-1} \bB_{12} \Delta \bX_{t-1}' &\cdots& \Delta\bX_{t-k} \bB_{k2}' \Sigma_2^{-1} \bB_{k2} \Delta \bX_{t-k}' &\Delta\bX_{t-k} \bB_{k2}' \Sigma_2^{-1}\\
    \Sigma_2^{-1} \bB_{12} \Delta \bX_{t-1}' &\cdots& \Sigma_2^{-1} \bB_{k2} \Delta \bX_{t-k}' &\Sigma_2^{-1}\\
    \end{pmatrix}.
\end{align*}
The full rank least squares estimator is given by $\tilde\bA_1 = \tilde\bS_{yx.z} \tilde\bS_{xx.z}^{-1}$ and $\tilde\Psi_1 = \tilde\bS_{yz} \tilde\bS_{zz}^{-1} - \tilde\bA_1 \tilde\bS_{xz} \tilde\bS_{zz}^{-1}$, where $\tilde\bS_{yx.z} = \tilde\bS_{yx} - \tilde\bS_{yz} \tilde\bS_{zz}^{-1} \tilde\bS_{zx}$, $\tilde\bS_{xx.z} = \tilde\bS_{xx} - \tilde\bS_{xz} \tilde\bS_{zz}^{-1} \tilde\bS_{zx}$. Denote $\tilde\bR_t := \Delta \bX_t  - \tilde{\bA}_1 \bX_{t-1} \bA_2' - \sum_{i=1}^{k} \tilde\bB_{i1} \Delta \bX_{t-i} \bB_{i2}' - \tilde\bD$. Let
\begin{align*}
    \tilde{\Sigma}_{\epsilon \epsilon} = \sum_t \sum_j \left(\by_{tj} - \tilde{\bA}_1 \bx_{tj} -  \tilde{\Psi}_1 \bz_{tj} \right)\left(\by_{tj} - \tilde{\bA}_1 \bx_{tj} -  \tilde{\Psi}_1 \bz_{tj} \right)' = \sum_t \tilde \bR_t \Sigma_2^{-1} \tilde \bR_t'.
\end{align*}
Take $\tilde{\bU}:=[\tilde{U}_1, \tilde{U}_1,\ldots,\tilde{U}_{r_1}]$. $\tilde{U}_j$ is the $j$-th leading unit eigenvector of $\tilde{\Sigma}_{\epsilon \epsilon}^{-1/2} \tilde{\bS}_{yx.z} \tilde{\bS}_{xx.z}^{-1} \tilde{\bS}_{xy.z} \tilde{\Sigma}_{\epsilon \epsilon}^{-1/2}$. Then $\bA_1$, $\Psi_1$, $\Sigma_1$ can be updated as
\begin{align*}
    \hat{\bA}_{1}^{\text{ml}} &= \tilde{\Sigma}_{\epsilon\epsilon}^{1/2} \tilde\bU \tilde\bU' \tilde{\Sigma}_{\epsilon\epsilon}^{-1/2} \tilde\bS_{yx} \tilde\bS_{yx}^{-1},\\
    \hat{\Psi}_1^{\text{ml}} &= \tilde\bS_{yz} \tilde\bS_{zz}^{-1} - \hat\bA_1 \tilde\bS_{xz} \tilde\bS_{zz}^{-1}, \\
    \hat{\Sigma}_1 &= \frac{1}{T-k-1} \sum_t \hat\bR_t \Sigma_2^{-1} \hat\bR_t',
\end{align*}
where $\hat\bR_t = \Delta \bX_t  - \hat{\bA}_1 \bX_{t-1} \bA_2' - \sum_{i=1}^{k} \hat\bB_{i1} \Delta \bX_{t-i} \bB_{i2}' -  \hat\bD$. It involves a joint estimation with stacked estimated parameters $\hat\Psi_1^{\text{ml}}$, such that the leading $d_1 \times d_1k$ sub-matrix is $\hat\bB_{11}, \ldots, \hat\bB_{k1}$, and the last $d_1 \times d_2$ sub-matrix is $\hat\bD$.
Given $\bA_1$, $\Sigma_1$ and $\bB_{i1}$, an update of $\bA_2$, $\Sigma_2$ and $\Psi_{2}$ can be similarly obtained.

\subsection{Least Squares Estimation}\label{sec:lse}
We denote the least squares estimator by $\hat\bA_i^{\text{ls}}$. Suppose $\bA_2, \bB_{12}, \ldots, \bB_{k2}$ are given. The alternating lease squares estimator of the model (\ref{model form 2}) minimizes the following residuals under the rank constraint $\rank(\bA_1)=r_1$:
\begin{equation}\label{minlse}
    \min_{\bA_1, \bB_{i1}, \bD, \text{rank}(\bA_1)=r_1} \sum_t\|\Delta \bX_t - \bA_1\bX_{t-1}\bA_2' - \sum_{i=1}^k  \bB_{i1}\bX_{t-i}\bB_{i2}' - \bD \|^2_F.
\end{equation}
Take $\hat{\bU}:=[\hat{U}_1, \hat{U}_1,\ldots,\hat{U}_{r_1}]$, where $\hat{U}_j$ is the $j$-th leading unit eigenvector of $\hat{\bS}_{yx.z} \hat{\bS}_{xx.z}^{-1} \hat{\bS}_{xy.z}$. Define $\Psi_1 = \{\bB_{11},\ldots,\bB_{k1},\bD\}$. Then $\hat{\bA}_1^{\text{ls}}$, $\hat{\Psi}_1^{\text{ls}}$ can be updated as
\begin{align*}
    \hat{\bA}_1^{\text{ls}} &=  \hat\bU \hat\bU' \hat\bS_{yx} \hat\bS_{yx}^{-1},\\
    \hat{\Psi}_1^{\text{ls}} &= \hat\bS_{yz} \hat\bS_{zz}^{-1} - \hat\bA_1 \hat\bS_{xz} \hat\bS_{zz}^{-1},
\end{align*}
where $\hat{\bS}_{xx} =  \sum_t \bX_{t-1} \bA_2' \bA_2 \bX_{t-1}'$, $\hat{\bS}_{yx} =  \sum_t \Delta \bX_{t} \bA_2 \bX_{t-1}'$, 
\begin{align*}
    \hat{\bS}_{yz} &= \sum_t \begin{pmatrix}
    \Delta\bX_{t} \bB_{12} \Delta \bX_{t-1}' & \cdots&
    \Delta\bX_{t} \bB_{k2} \Delta \bX_{t-k}'
    & \Delta\bX_{t} 
    \end{pmatrix}, \\
    \hat{\bS}_{yz} &= \sum_t \begin{pmatrix}
    \Delta\bX_{t} \bB_{12} \Delta \bX_{t-1}' & \cdots & 
    \Delta\bX_{t} \bB_{k2} \Delta \bX_{t-k}'
    & \Delta\bX_{t} 
    \end{pmatrix}, \\
    \hat{\bS}_{xz} &= \sum_t \begin{pmatrix}
    \bX_{t-1} \bA_{2}'  \bB_{12} \Delta \bX_{t-1}' & \cdots & \bX_{t-1} \bA_{2}'  \bB_{k2} \Delta \bX_{t-k}' & \bX_{t-1} \bA_{2}' 
    \end{pmatrix}, \\
    \hat{\bS}_{zz} &= 
    \sum_t \begin{pmatrix}
    \Delta\bX_{t-1} \bB_{12}' \bB_{12} \Delta \bX_{t-1}' &\cdots& \Delta\bX_{t-1} \bB_{12}'  \bB_{k2} \Delta \bX_{t-k}' &\Delta\bX_{t-1} \bB_{12}' \\
    \cdots & \cdots & \cdots & \cdots \\
    \Delta\bX_{t-k} \bB_{k2}  \bB_{12} \Delta \bX_{t-1}' &\cdots& \Delta\bX_{t-k} \bB_{k2}'  \bB_{k2} \Delta \bX_{t-k}' &\Delta\bX_{t-k} \bB_{k2}' \\
    \bB_{12} \Delta \bX_{t-1}' &\cdots&  \bB_{k2} \Delta \bX_{t-k}' & \bI_{d_2}\\
    \end{pmatrix}.
\end{align*}
Given $\bA_1$ and $\bB_{i1}$, an update of $\bA_2$ and $\Psi_2$ can be similarly obtained.

\section{Asymptotics}\label{SEC:Theorems}
In this section, we give the asymptotic distribution of the LSE and MLE estimators of the general form (\ref{model form 2}), where denoted as $\hat\bbeta^{\text{ls}}$ and $\hat\bbeta^{\text{ml}}$. We first establish the central limit theorem of $\hat\balpha$, $\hat\bB_i$ and $\hat\bD$. Then we derive asymptotic results for $\tilde\bbeta_{i} = \hat\bbeta_i (\bbeta_i' \hat\bbeta_i)^{-1}$, such that $\tilde\bbeta$ is normalized by $\bbeta_{i}' \tilde\bbeta = \bI_{r_i}$. We then extend this to the estimators $\hat\bbeta_{ic} = \hat\bbeta_i (\bc' \hat\bbeta_i)^{-1}$ with any normalizing matrix $\bc$ that $\bc' \hat\bbeta_{ic} = \bI_{r_i}$. We also establish the asymptotic distribution for the projection matrix $\hat{\mathbb{P}}_{\beta} = \hat\bbeta_c (\hat\bbeta_c' \hat\bbeta_c)^{-1} \hat\bbeta_c'$. The asymptotic analysis involve heavy notations. First of all, we make the convention that $\bbeta_i'\bbeta_i = \hat\bbeta_i'\hat\bbeta_i = \bI_{r_i}$, and $\|\bA_1\|_F=\|\hat\bA_1\|_F=1$, $i=1,2$. In the table below, we list notations that appear in Theorems of both LSE and MLE estimators.
\begin{table}[!ht]
\centering
\begin{tabular}{l|l}
\hline
Notations &  \\ \hline
$\bd$  & $\vec{(\bD)}$ \\
$\h{\delta}_1$  & $(\vec{(\balpha_1)}', \h{0}')'$ \\
$\h{\delta}_i$  &  $(\h{0}_{d_1r_1+id_2r_2}', \vec{(\bB_{i1})}', \h{0}')'$ for $i=2,\ldots,k$ \\
$\bC$ & $\bbeta_{\bot} (\balpha_{\bot}' \Gamma \bbeta_{\bot})^{-1} \balpha_{\bot}'$, where $\Gamma = \bI - \sum_{i=1}^{k}\bB_{i2} \otimes \bB_{i1}$\\
$\h\theta$ & $\{\balpha_1, \balpha_2', \bB_{11}, \bB_{12}',\ldots, \bB_{k1}, \bB_{k2}'\}$ \\
$\tilde\bbeta_{i}$  & the estimator normalized by $\bbeta_i$ such that $\tilde\bbeta_{i} = \hat\bbeta_{i} (\bbeta_i'\hat\bbeta_{i})^{-1}$ \\
$\hat\bbeta_{ic}$  & the estimator normalized by $\bc$ such that $\hat\bbeta_{i} = \hat\bbeta_{i} (\bc'\hat\bbeta_{i})^{-1}$ \\
$\left(\h{\tau}_1, \h{\gamma}_1, \bbeta_1 \right)$ & $\h{\tau}_1 = \bU_1$, and $\h{\gamma}_1$ orthogonal to $\h{\tau}_1$ and $\bbeta_1$, where \\
 & $\bU_1 \Lambda_1 \bV_1' = [\vec^{-1} (\bC \bd)] \bbeta_2$ is the SVD decomposition \\
 $\bB_{Ti}$ & $(\bar{\h{\gamma}}_i, T^{-\frac{1}{2}} \bar{\h{\tau}}_i)$, where $\bar{\h{\gamma}}_i = \h{\gamma}_i(\h{\gamma}_i' \h{\gamma}_i)^{-1}$, $\bar{\h{\tau}}_i = \h{\tau}_i(\h{\tau}_i' \h{\tau}_i)^{-1}$, $i=1,2$\\
 $\bU_{Ti}$ & $(\h{\gamma}_i, T^{\frac{1}{2}}\h{\tau}_i)' \tilde \bbeta_i$, $i=1,2$\\
\hline
\end{tabular}
\end{table}
Define 
\begin{align*}
\boldsymbol{W}_{t-1} = \begin{pmatrix} 
\bbeta_1' (\bX_{t-1} - \bar{\bX}_{t-1}) \bbeta_2 \balpha_2' \otimes \bI_{d_1}\\
\bI_{d_2} \otimes \bbeta_2' (\bX_{t-1} - \bar{\bX}_{t-1})' \bbeta_1 \balpha_1'\\
(\Delta\bX_{t-i} - \Delta\bar{\bX}_{t-i})  \bB_{i2} \otimes \bI_{d_1} \\
\bI_{d_2} \otimes (\Delta\bX_{t-i} - \Delta\bar{\bX}_{t-i})' \bB_{i1} \\
i=1,\ldots,k 
\end{pmatrix},
\end{align*}
where $\bar{\bX_t} = \sum_{t=1}^{T} \bX_t / T$. Then define $\bH_1 = \mathbb{E} \left( \bW_t \bW_t'\right) + \sum_{i=1}^{k}\h{\delta}_i \h{\delta}_i'$, and $\bH_2 = \mathbb{E} \left( \bW_t \Sigma^{-1} \bW_t' \right) + \sum_{i=1}^{k}\h{\delta}_i \h{\delta}_i'$,  where $\h{\delta}_1 = (\vec{(\balpha_1)}', \h{0}')'$, and $\h{\delta}_i = (\h{0}_{d_1r_1+id_2r_2}', \vec{(\bB_{i1})}', \h{0}')'$ for $i=2,\ldots,k$.

\subsection{Asymptotics for matrix I(1) processes}
The probability properties of a matrix I(1) process provide the foundation for analyzing the asymptotic distribution of proposed estimators. In this section, we extend some classical results for vector I(1) processes to matrix I(1) processes.
By the Representation Theorem~\ref{Theorem:represent}, after vectorization, the process (\ref{model form 2}) can be expressed as a combination of a random walk, a linear trend and a stationary process:
\begin{align}\label{vecProcess}
    \vec{(\bX_t)} = \bC \sum_{i=1}^{t} \bepsilon_i + \bC \bd t + C(L)(\bepsilon_t + \bd) + \mathbb{P}_{\beta_{\bot}}\vec{(\bX_0)},
\end{align}
where $\bepsilon_i = \vec{(\bE_i)}$, $\bd = \vec{(\bD)}$, $\bC = \bbeta_{\bot} (\balpha_{\bot}' \Gamma \bbeta_{\bot})^{-1} \balpha_{\bot}'$, $\Gamma = \bI - \sum_{i=1}^{k}\bB_{i2} \otimes \bB_{i1}$, and $C(L)(\bepsilon_t)$ is the stationary part. Denote $\bw(u)$, $u \in [0,1]$, as the $d$-dimensional Brownian motion such that $\bw(t)$ is $N_d(0, t \Sigma)$. Then we have
$$T^{-\frac{1}{2}} \sum_{i=1}^{[Tu]} \bepsilon_i \overset{w}{\longrightarrow} \bw(u), u\in[0,1].$$
Denote $\bW(u) = \vec^{-1}\left(\bC \bw(u)\right) \in \mathbb{R}^{d_1 \times d_2}$.
By the formula of $\bC$, we have $\col(\left[\vec^{-1}(\bC\bd)\right] \bbeta_2) \subseteq \col(\bbeta_{1\bot})$. If $\rank(\left[\vec^{-1}(\bC\bd)\right] \bbeta_{2}) = k_1$, the Singular Value Decomposition (SVD)
\[\left[\vec^{-1}(\bC\bd)\right] \bbeta_{2} = \bU_1 \h\Lambda_1 \bV_1',\]
in which $\h\Lambda_1$ is a square diagonal matrix of size $k_1 \times k_1$, $\bU_1$ is a $d_1 \times k_1$ semi-unitary matrix. Then we can take $\h{\tau}_1 = \bU_1$ and $\h{\gamma}_1 \in \mathbb{R}^{d_1 \times (d_1-r_1-k_1)}$ be orthogonal to $\h{\tau}_1$ and $\bbeta_1$, such that $\bar{\h{\gamma}}_1\h{\gamma}_1' + \bar{\h{\tau}}\h{\tau}' = \bar\bbeta_{1\bot}\bbeta_{1\bot}'$.
Then the estimator
\begin{align}
\tilde\bbeta_i = \hat\bbeta_i (\bbeta_i' \hat\bbeta_i)^{-1} = \bbeta_i + \bB_{Ti} \bU_{Ti},
\end{align}
where $\bB_{Ti} = (\bar{\h{\gamma}}_i, T^{-\frac{1}{2}} \bar{\h{\tau}}_i)\in\mathbb{R}^{d_i \times (d_i - r_i)}$, $\bU_{Ti} =  (\h{\gamma}_i, T^{\frac{1}{2}}\h{\tau}_i)' \tilde \bbeta_i \in \mathbb{R}^{(d_i-r_i) \times r_i}$, $i=1,2$.

\begin{lemma}\label{G}
Let $\h{\tau}_1 = \bU_1 \in \mathbb{R}^{d_1 \times k_1}$ and $\h{\gamma}_1 \in \mathbb{R}^{d_1 \times (d_1 - r_1 - k_1)}$ be orthogonal to $\h{\tau}_1$ and $\bbeta_1$ where $\left[\vec^{-1}(\bC\bd)\right] \bbeta_{2} = \bU_1 \h\Lambda_1 \bV_1'$ is the compact SVD decomposition. Then $(\bbeta_1, \h{\tau}_1, \h{\gamma}_1)$ has full rank $d_1$. As $T \to \infty$ and $u \in [0,1]$, we have
\begin{align}
    T^{-\frac{1}{2}} \bar{\h{\gamma}}_1' \bX_{[Tu]} \bbeta_2  &\to \bar{\h{\gamma}}_1' \bW(u) \bbeta_2 \in \mathbb{R}^{(d_1 -r_1 - k_1) \times r_2}, \\
    T^{-1} \bar{\h{\tau}}'_1 \bX_{[Tu]}  \bbeta_2 &\to \h\Lambda_1 \bV_1' u \in \mathbb{R}^{k_1 \times r_2}.
\end{align}
Define $\bB_{T1} = (\bar{\h{\gamma}}_1, T^{-\frac{1}{2}} \bar{\h{\tau}}_1)$ then we have
\[T^{-\frac{1}{2}} \bB_{T1}' \bX_{[Tu]} \bbeta_2 \to \begin{pmatrix} 
\bar{\h{\gamma}}_1' \bW(u) \bbeta_2  \\
\h\Lambda_1 \bV_1' u\\
\end{pmatrix} \in \mathbb{R}^{(d_1 - r_1) \times r_2}.\]
It follows that
\begin{align}
    T^{-\frac{1}{2}} \bB_{T1}' (\bX_{[Tu]} - \bar\bX_T)  \bbeta_2 \to  \begin{pmatrix} 
\bar{\h{\gamma}}_1' (\bW(u) - \bar\bW) \bbeta_2 \\
\h\Lambda_1 \bV_1' (u - \frac{1}{2})\\
\end{pmatrix}
\end{align}
where $\bar{\bW} = \int_{0}^{1} \bW(u) du$.
\end{lemma}
The proof is given in Appendix~\ref{Appendix:proofs}. Thus, the asymptotic properties of the process depend on which linear combination of the process we consider. If we examine the processes $\h{\tau}'_1 \bX_t \bbeta_2$, $\h{\gamma}'_1 \bX_t \bbeta_2$, and $\bbeta'_1 \bX_t \bbeta_2$ separately, we observe different types of dominant terms. In particular, for $\h{\tau}'_1 \bX_t \bbeta_2$, the process is dominated by the linear trend, while for $\h{\gamma}'_1 \bX_t \bbeta_2$, the dominating term is the random walk part. If we consider linear combinations of the form $\bbeta'_1 \bX_t \bbeta_2$, the process becomes stationary.

To introduce the asymptotics of estimators, we introduce some notations following the results of Lemma~\ref{G}.
\begin{equation}
    \bG_1(u) := \left(\begin{pmatrix} 
\bar{\h{\gamma}}_1' (\bW(u) - \bar\bW) \bbeta_2 \\
\h\Lambda_1 \bV_1' (u - \frac{1}{2})\\
\end{pmatrix} \balpha_2' \right)\otimes \balpha_1' \in \mathbb{R}^{(d_1 - r_1)r_1 \times d_1d_2},
\end{equation}
such that
\[
T^{-\frac{1}{2}} \bB_{T1}' (\bX_{[Tu]} - \bar\bX_T)  \bA_2 \otimes \balpha_1' \to
\bG_1(u).
\]
Similarly, we can define $\bG_2(u)$ and we combined the rows of $\bG_1$, $\bG_2$ to make $\bG$,
\[
\bG(u) :=
 \begin{pmatrix} 
\bG_1(u)\\
\bG_2(u)
\end{pmatrix}
\ \text{of size} \
\begin{pmatrix}
(d_1-r_1)r_1 \times d_1d_2 \\
(d_2-r_2)r_2 \times d_1d_2 
\end{pmatrix}.
\]

\subsection{Asymptotics for LSE estimators}\label{CMAR:TheoremsLSE}

\begin{theorem}\label{Theorem:LSEtheta}
Assume that $\{\bE_t\}$ are IID with mean zero and finite second moments. Assume that $0<\rank\bA_i = r_i$, and the Assumption~\ref{rootA} holds. Then the asymptotic distribution of the LSE estimator $\hat{\h\theta}^{\text{ls}}$ is given by
\begin{align}\label{limit:LSEtheta}
    \sqrt{T} \vec{(\hat{\h{\theta}}^{\text{ls}} - \h{\theta})} \to N(\h{0}, \Xi_1),
\end{align}
where $\Xi_1 = \bH_1^{-1}\mathbb{E}(\bW_t \Sigma \bW_t')\bH_1^{-1}$, $\h\theta = \{\balpha_1, \balpha_2', \bB_{11}, \bB_{12}',\ldots, \bB_{k1}, \bB_{k2}'\}$.
\end{theorem}

The proof is given in Appendix~\ref{Appendix:proofs}. The parameters not involving cointegrating vectors and constant terms are asymptotically normal under a convergence rate of $O_p(1/\sqrt{T})$. Theorem~\ref{Theorem:LSEtheta} includes the CVAR model considered \cite{johansen1995likelihood} as a special case with $d_2=1$. The asymptotic properties of the constant term $\hat\bd = \vec(\hat{\bD})$ can be obtained directly from \cite{johansen1995likelihood}, and we put it in Appendix~\ref{Appendix:theorems}.

\begin{theorem}\label{Theorem:LSEtildebeta}
Assume the same conditions as Theorem~\ref{Theorem:LSEtheta}. In addition, assume the error matrix $\bE_t$
are IID normal. It holds that
\begin{align}\label{Formula:LSEtildebeta}
\begin{bmatrix}
     \vec\left((\tilde\bbeta_1^{\text{ls}} - \bbeta_1)' (T\h{\gamma}_1, T^{3/2}\h{\tau}_1)\right) \\
     \vec\left((T\h{\gamma}_2, T^{3/2}\h{\tau}_2)' (\tilde\bbeta_2^{\text{ls}} - \bbeta_2)\right)
\end{bmatrix}
\overset{w}{\longrightarrow}
    \left(\int_0^1 \bG\bG' \mathrm{d}u \right)^{-1} \int_0^1 \bG \mathrm{d}\bW.
\end{align}
\end{theorem}

\begin{corollary}\label{Corollary:LSEtildebeta}
Under the assumptions of Theorem~\ref{Theorem:LSEtildebeta}, it holds that
\begin{align}
    \vec\left((\tilde\bbeta_1^{\text{ls}} - \bbeta_1)' (T\h{\gamma}_1, T^{3/2}\h{\tau}_1)\right) \overset{w}{\longrightarrow}
    \left(\int_0^1 \bG_{1.2}\bG_{1.2}' \mathrm{d}u \right)^{-1} \int_0^1 \bG_{1.2} \mathrm{d}\bW, \label{Formula:tildebeta1}\\
    \vec\left((T\h{\gamma}_2, T^{3/2}\h{\tau}_2)' (\tilde\bbeta_2^{\text{ls}} - \bbeta_2)\right) \overset{w}{\longrightarrow}
    \left(\int_0^1 \bG_{2.1}\bG_{2.1}' \mathrm{d}u \right)^{-1} \int_0^1 \bG_{2.1} \mathrm{d}\bW.\label{Formula:tildebeta2}
\end{align}
where $\bG_{i.j}(u) = \bG_i(u) - \int_{0}^{1} \bG_i \bG_j' du \left[ \int_{0}^{1} \bG_j \bG_j' du \right]^{-1} \bG_j(u)$.
\end{corollary}

The proof of Theorem~\ref{Theorem:LSEtildebeta} and Corollary~\ref{Corollary:LSEtildebeta} are given in Appendix~\ref{Appendix:proofs}. It is shown above $\tilde\bbeta_i$ is super-consistent that $\tilde\bbeta_i - \bbeta_i = O_p(1/T)$. And note that the speed of convergence is different in the different directions $\h{\tau}_i$ and $\h{\gamma}_i$ corresponding to different behaviors of the process $\h{\tau}'_i\bX_t \bbeta_2$ and $\h{\gamma}'_i\bX_t \bbeta_2$ as given in Lemma~\ref{G}. As we known $\tilde{\bbeta}_i$ is normalized by $\bbeta_i$ such that $\bbeta'_i \tilde{\bbeta}_i = \bI_{r_i}$. Next, we generalize the results to the estimators $\hat\bbeta_{ic} = \hat\bbeta_i (\bc' \hat\bbeta_i)^{-1}$ with any normalizing matrix $\bc$ that $\bc' \hat\bbeta_{ic} = \bI_{r_i}$.

\begin{theorem}\label{Theorem:LSEhatbeta}
Assume the same conditions as Theorem~\ref{Theorem:LSEtildebeta}. Let $\hat{\bbeta}_{i}^{\text{ls}} = \hat{\bbeta}_i^{\text{ls}} (\bc_i' \hat{\bbeta}_i^{\text{ls}})^{-1}$, $\bbeta_{ic} = \bbeta_i (\bc_i' \bbeta_i)^{-1}$, such that $\bc_i' \bbeta_{ic} = \bc_i' \hat{\bbeta}_{i}^{\text{ls}} = \bI_{r_i}$, $i=1,2$. Then we have
\begin{align}
T\vec{\left((\hat{\bbeta}_{1}^{\text{ls}})' - \bbeta_{1c}'\right)}
    \overset{w}{\longrightarrow}
    \left(\left((\bI_{d_1} - \bbeta_{1c} \bc_1') \bar{\h{\gamma}}_{1} \right)\otimes \bI_{r_1} \right)\left(\int_0^1 \bQ_{1.2} \bQ_{1.2}' \mathrm{d}u \right)^{-1} \int_0^1 \bQ_{1.2} \mathrm{d}\bW, \\
T\vec{\left(\hat{\bbeta}_{2}^{\text{ls}} - \bbeta_{2c}\right)}
    \overset{w}{\longrightarrow}
    \left(\bI_{r_2} \otimes \left((\bI_{d_2} - \bbeta_{2c} \bc_1') \bar{\h{\gamma}}_{2} \right)\right)\left(\int_0^1 \bQ_{2.1} \bQ_{2.1}' \mathrm{d}u \right)^{-1} \int_0^1 \bQ_{2.1} \mathrm{d}\bW.
\end{align}
where $\bQ_{i.j}(u) = \bG_{i.j,1}(u) - \int_{0}^{1} \bG_{i.j,1} \bG_{i.j,2}' du \left[ \int_{0}^{1} \bG_{i.j,2} \bG_{i.j,2}' du \right]^{-1} \bG_{i.j,2}(u)$, and $\bG_{i.j,1}$, $\bG_{i.j,2}$ are block sub-matrices such that $\bG_{i.j}' = \{\bG_{i.j,1}', \bG_{i.j,2}'\}$.
\end{theorem}

Furthermore, we examine the estimated cointegration space, represented by the projection matrix $\hat{\mathbb{P}}_{\beta_i}^{\text{ls}} = \hat\bbeta^{\text{ls}}_i \left((\hat\bbeta^{\text{ls}}_i)' \hat\bbeta^{\text{ls}}_i \right)^{-1}(\hat\bbeta^{\text{ls}}_i)'$, which is unique regardless of the choice of normalization. Note the true cointegration space $\mathbb{P}_{\beta_i} = \bbeta_i (\bbeta_i' \bbeta_i)^{-1} \bbeta_i' = \bbeta_i \bbeta_i'$, since we require $\bbeta_i' \bbeta_i = \bI_{r_i}$, $i=1,2$.
\begin{theorem}\label{Theorem:LSEbb}
Assume the same conditions as Theorem~\ref{Theorem:LSEtildebeta}. It holds that
\begin{align}
\begin{split}
    &T
    \vec\left(  (\hat{\mathbb{P}}_{\beta_1}^{\text{ls}} - \mathbb{P}_{\beta_1})(\h{\gamma}_1, T^{1/2}\h{\tau}_1) \right)
    \overset{w}{\longrightarrow} \\
    &\left(\bI_{d_1} \otimes \bbeta_1 + (\bbeta_1 \otimes \bI_{d_1})\bP_{d_1-r_1, r_1}\right) \left(\int_0^1 \bG_{1.2}\bG_{1.2}' \mathrm{d}u \right)^{-1} \int_0^1 \bG_{1.2} (\mathrm{d}\bW),
\end{split}
\end{align}
where $\bP_{d_1-r_1, r_1}$ is the permutation matrix such that $\vec(\bbeta_{1\bot}) = \bP_{d_1-r_1, r_1} \vec(\bbeta_{1\bot}')$.
\end{theorem}

\subsection{Asymptotics for MLE estimators}\label{CMAR:TheoremsMLE}

\begin{theorem}\label{Theorem:MLEtheta}
Assume that $\{\bE_t\}$ are IID with mean zero and finite second moments. Assume the Assumption~\ref{rootA} holds, and $0<\rank\bA_i = r_i$, and $\Sigma_e$ takes the form~(\ref{sigma_e}), and is non-singular. Then the asymptotic distribution of the estimator $\hat{\h\theta}^{\text{ml}}$ is given by
\begin{align}\label{MLEtheta}
    \sqrt{T} \vec{(\hat{\h{\theta}}^{\text{ml}} - \h{\theta})} \to N(\bzero, \Xi_2),
\end{align}
where $\Xi_2 = \bH_2^{-1}\mathbb{E}(\bW_t \Sigma^{-1} \bW_t')\bH_2^{-1}$, $\bH_2 = \mathbb{E} \left( \bW_t \Sigma^{-1} \bW_t' \right) + \sum_{i=1}^{k}\h{\delta}_i \h{\delta}_i'$. And we defined
$\h{\theta}$, $\bW_t$, and $\h{\delta}_i$ in Theorem~\ref{Theorem:LSEtildebeta}.
\end{theorem}

\begin{theorem}\label{Theorem:MLEtildebeta}
Assume the same conditions as Theorem~\ref{Theorem:MLEtheta}. It holds that
\begin{align}\label{Formula:MLEtildebeta}
\begin{bmatrix}
     \vec\left((\tilde\bbeta_1^{\text{ml}} - \bbeta_1)' (T\h{\gamma}_1, T^{3/2}\h{\tau}_1)\right) \\
     \vec\left((T\h{\gamma}_2, T^{3/2}\h{\tau}_2)' (\tilde\bbeta_2^{\text{ml}} - \bbeta_2)\right)
\end{bmatrix}
\overset{w}{\longrightarrow}
    \left(\int_0^1 \bG \Sigma^{-1}\bG' \mathrm{d}u \right)^{-1} \int_0^1 \bG \Sigma^{-1}\mathrm{d}\bW.
\end{align}
\end{theorem}

\begin{corollary}
Under the same assumptions of Theorem~\ref{Theorem:MLEtildebeta}, it holds that
\begin{align*}
    \vec\left((\tilde\bbeta_1^{\text{ml}} - \bbeta_1)' (T\h{\gamma}_1, T^{3/2}\h{\tau}_1)\right) \overset{w}{\longrightarrow}
    \left(\int_0^1 \tilde\bG_{1.2} \Sigma^{-1} \tilde\bG_{1.2}' \mathrm{d}u \right)^{-1} \int_0^1 \tilde\bG_{1.2} \Sigma^{-1} \mathrm{d}\bW,\\
    \vec\left((T\h{\gamma}_2, T^{3/2}\h{\tau}_2)' (\tilde\bbeta_2^{\text{ml}} - \bbeta_2)\right) \overset{w}{\longrightarrow}
    \left(\int_0^1 \tilde\bG_{2.1} \Sigma^{-1} \tilde\bG_{2.1}' \mathrm{d}u \right)^{-1} \int_0^1 \tilde\bG_{2.1} \Sigma^{-1} \mathrm{d}\bW.
\end{align*}
where $\tilde\bG_{i.j}(u) = \bG_i(u) - \int_{0}^{1} \bG_i \Sigma^{-1}\bG_j' du \left[ \int_{0}^{1} \bG_j \Sigma^{-1} \bG_j' du \right]^{-1} \bG_j(u)$.
\end{corollary}

\begin{theorem}\label{Theorem:MLEhatbeta}
Assume the same conditions as Theorem~\ref{Theorem:MLEtheta}. Let $\hat{\bbeta}_i^{\text{ml}} = \hat{\bbeta}_i^{\text{ml}} (\bc_i' \hat{\bbeta}_i^{\text{ml}})^{-1}$, such that $\bc_i' \hat{\bbeta}_{i}^{\text{ml}} = \bI_{r_i}$. 
Then we have
\begin{align*}
T\vec{\left((\hat{\bbeta}_1^{\text{ml}})' - \bbeta_{1c}'\right)}
    \overset{w}{\longrightarrow}
    \left(\left((\bI_{d_1} - \bbeta_{1c} \bc_1') \bar{\h{\gamma}}_{1} \right)\otimes \bI_{r_1} \right)\left(\int_0^1 \tilde\bQ_{1.2} \Sigma^{-1} \tilde\bQ_{1.2}' \mathrm{d}u \right)^{-1} \int_0^1 \tilde\bQ_{1.2} \Sigma^{-1} \mathrm{d}\bW, \\
T\vec{\left(\hat{\bbeta}_2^{\text{ml}} - \bbeta_{2c}\right)}
    \overset{w}{\longrightarrow}
    \left(\bI_{r_2} \otimes \left((\bI_{d_2} - \bbeta_{2c} \bc_2') \bar{\h{\gamma}}_{2} \right)\right)\left(\int_0^1 \tilde\bQ_{2.1} \Sigma^{-1} \tilde\bQ_{2.1}' \mathrm{d}u \right)^{-1} \int_0^1 \tilde\bQ_{2.1} \Sigma^{-1} \mathrm{d}\bW.
\end{align*}
where $\tilde\bQ_{i.j}(u) = \tilde\bG_{i.j,1}(u) - \int_{0}^{1} \tilde\bG_{i.j,1} \Sigma^{-1} \tilde\bG_{i.j,2}' du \left[ \int_{0}^{1} \tilde\bG_{i.j,2} \Sigma^{-1} \tilde\bG_{i.j,2}' du \right]^{-1} \tilde\bG_{i.j,2}(u)$, and $\tilde\bG_{i.j,1}$, $\tilde\bG_{i.j,2}$ are block sub-matrices such that $\tilde\bG_{i.j}' = \{\tilde\bG_{i.j,1}', \tilde\bG_{i.j,2}'\}$.
\end{theorem}

\begin{theorem}\label{Theorem:MLEbb}
Assume the same conditions as Theorem~\ref{Theorem:MLEtheta}. It holds that
\begin{align}
\begin{split}
    &T
    \vec\left(  (\hat{\mathbb{P}}_{\beta_1}^{\text{ml}} - \mathbb{P}_{\beta_1})(\h{\gamma}_1, T^{1/2}\h{\tau}_1) \right)
    \overset{w}{\longrightarrow} \\
    &\left(\bI_{d_1} \otimes \bbeta_1 + (\bbeta_1 \otimes \bI_{d_1})\bP_{d_1-r_1, r_1}\right) \left(\int_0^1 \tilde\bG_{1.2} \Sigma^{-1} \tilde\bG_{1.2}' \mathrm{d}u \right)^{-1} \int_0^1 \tilde\bG_{1.2} \Sigma^{-1} \mathrm{d}\bW, \label{Formula:mletildebeta1}
\end{split}
\end{align}
where $\bP_{d_1-r_1, r_1}$ is the permutation matrix such that $\vec(\bbeta_{1\bot}) = \bP_{d_1-r_1, r_1} \vec(\bbeta_{1\bot}')$.
\end{theorem}

\section{Numerical Results}\label{SEC:simulation}

\subsection{Simulations}

In this section, we investigate the empirical performance of the proposed estimators under various simulation setups. The simulations compare the proposed alternating least square estimator $\hat\beta^{ls}$ (LSE), the maximum likelihood estimator $\hat\bbeta^{ml}$ (MLE), and the estimator of the cointegrated vector autoregressive model (CVAR) as a benchmark. 

The true ranks are taken as known. Given dimensions $d_i$ and ranks $r_i$, $i=1,2$, the observed data $\bX_t$ are simulated according to the model (\ref{model form 2}) with $k=1$. The coefficient matrix $\bA_1$ is generated according to $\bA_1 = \bQ_1 \Lambda \bQ_2'$, where the elements of the $r_1 \times r_1$ diagonal matrix $\Lambda$ are IID absolute standard normal random variables, and the $d_1 \times r_1$ orthogonal matrix $\bQ_1$ and $\bQ_2$ are randomly generated from the Haar measure. The matrix $\bA_2$ is generated in the same way. The matrices $\bB_{i1}$ are generated in a similar way except that $\Lambda$, $\bQ_1$ and $\bQ_2$ are $d_i \times d_i$ matrices. Note that $\rank\bA_j = r_j$, so we have the decomposition $\bA_i = \balpha_i \bbeta_i'$. Denote $\bbeta = \bbeta_2 \otimes \bbeta_1$, $\balpha = \balpha_2 \otimes \balpha_1$, $\bB = \bB_2 \otimes \bB_1$. After vectorization, the model~(\ref{vecCMAR}) without the constant term can be written in a compact form
\[\begin{pmatrix}
    \bbeta' \bX_t \\
    \Delta \bX_t
\end{pmatrix} =
\Phi
\begin{pmatrix}
    \bbeta' \bX_{t-1} \\
    \Delta \bX_{t-1}
\end{pmatrix}
+
\begin{pmatrix}
    \bbeta' \bE_{t} \\
    \bE_{t}
\end{pmatrix}, \ \text{where} \ \Phi := \begin{pmatrix}
    (\bI_{r_1r_2} + \bbeta' \balpha) & \bbeta' \bB \\
    \balpha & \bB
\end{pmatrix}.
\]
To generate an I(1) process, we set the spectral radius $\rho(\Phi) < 1$. The model with and without the constant term are both studied. If included the constant term, $\bD$ is generated randomly with entries are IID standard normal random variables, and $\|\bD\|_F = 0.8$. The error tensors $\bE_t$ are IID normal with covariance matrix $\Sigma_e:= \mathrm{Cov}[\vec(\mathcal{E}_t)]$, for which we consider two settings:

\begin{itemize}
    \item[(I)] $\Sigma_e = \bQ \Lambda \bQ'$, where the elements of the diagonal matrix $\Lambda$ are equally spaced over the interval $[1, 10]$, and $\bQ$ is a random orthogonal matrix generated from the Haar measure.
    \item[(II)] $\Sigma_e$ takes the form $\Sigma_e = \Sigma_2 \otimes \Sigma_1$, where $\Sigma_1$ and $\Sigma_2$ are generated similarly as the $\Sigma_e$ in Setting~I, except that the diagonal entries of $\Lambda$ are equally spaced over $[1, 5]$.
\end{itemize}
For each configuration of sample size $T$, dimensions $d_i$ and ranks $r_i$, we repeat the simulation 100 times, and show the box plots of the estimation error
\begin{equation*}
    \log \left(\|\hat{\mathbb{P}}_{\beta_1} - \mathbb{P}_{\beta_1}\|_s^2 + \|\hat{\mathbb{P}}_{\beta_2} - \mathbb{P}_{\beta_2}\|_s^2 \right),
\end{equation*}
where $\mathbb{P}_{\beta_i} = \bbeta_i (\bbeta_i' \bbeta_i)^{-1} \bbeta_i'$ is the projection matrix of $\bbeta_i$. For a particular simulation setting with multiple repetitions, the coefficient matrices $\bA_i$, $\bB_{1i}$, $i=1,2$. We consider the model with the constant term, under the two settings of $\Sigma_e$, represented by Figure~\ref{Figure:CMARSVD2} and Figure~\ref{Figure:CMARMLE2}, respectively. Also, Figure~\ref{Figure:CMARSVD1} under Setting~I and Figure~\ref{Figure:CMARMLE1} under Setting~II are simulated without the constant term.

\begin{figure}[!ht]
    \centering
    \includegraphics[width = \textwidth]{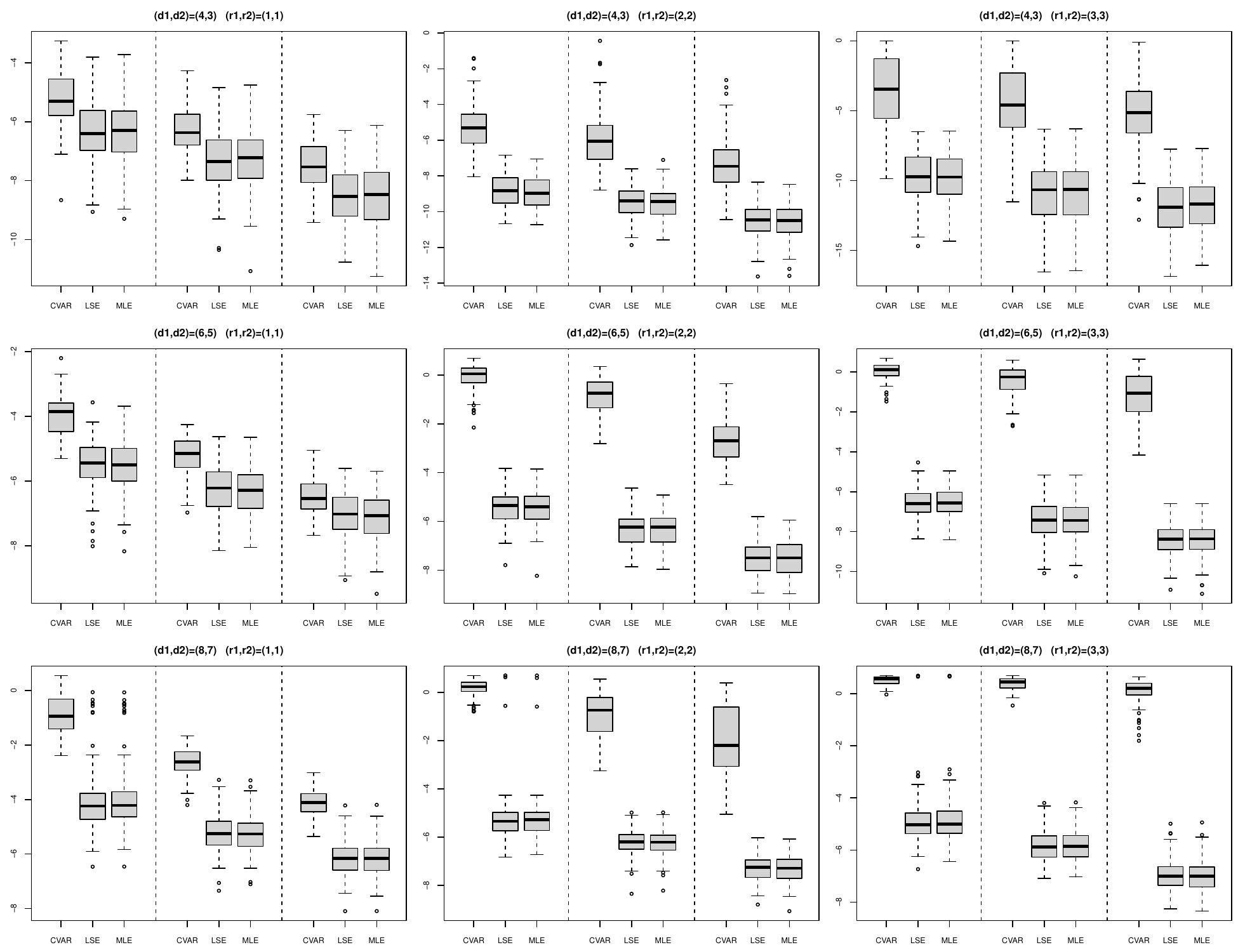}
    \caption{Comparison of CVAR, LSE and MLE. True model is the model~\ref{model form 2}, $k=1$, without the constant term, under setting~I. The three panels in each figure correspond to sample sizes 400, 600 and 1000 respectively.}
    \label{Figure:CMARSVD1}
\end{figure}

\begin{figure}[!ht]
    \centering
    \includegraphics[width = \textwidth]{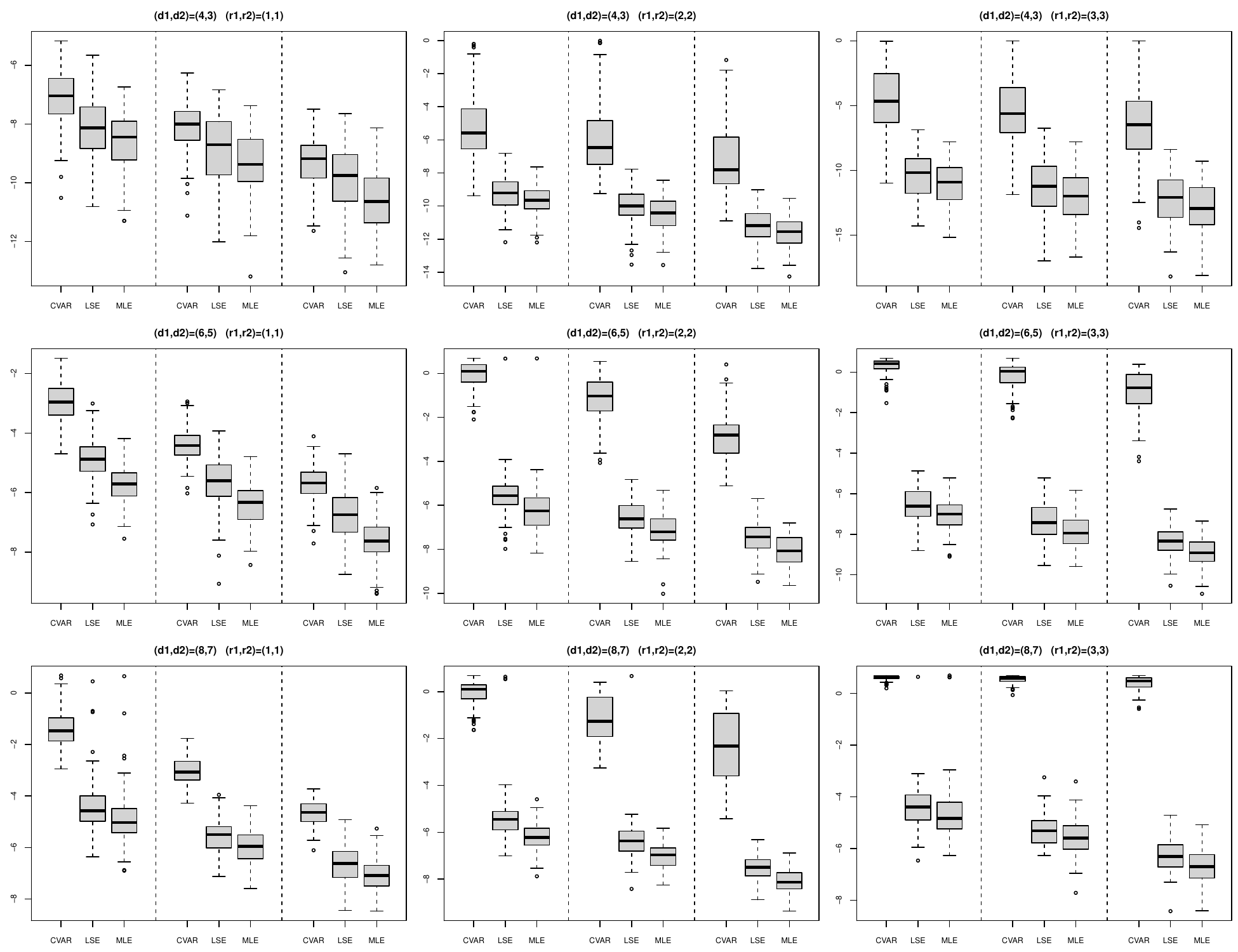}
    \caption{Comparison of CVAR, LSE and MLE. True model is the model~\ref{model form 2}, $k=1$, without the constant term, under setting~II. The three panels in each figure correspond to sample sizes 400, 600 and 1000 respectively.}
    \label{Figure:CMARMLE1}
\end{figure}

\begin{figure}[!ht]
    \centering
    \includegraphics[width = \textwidth]{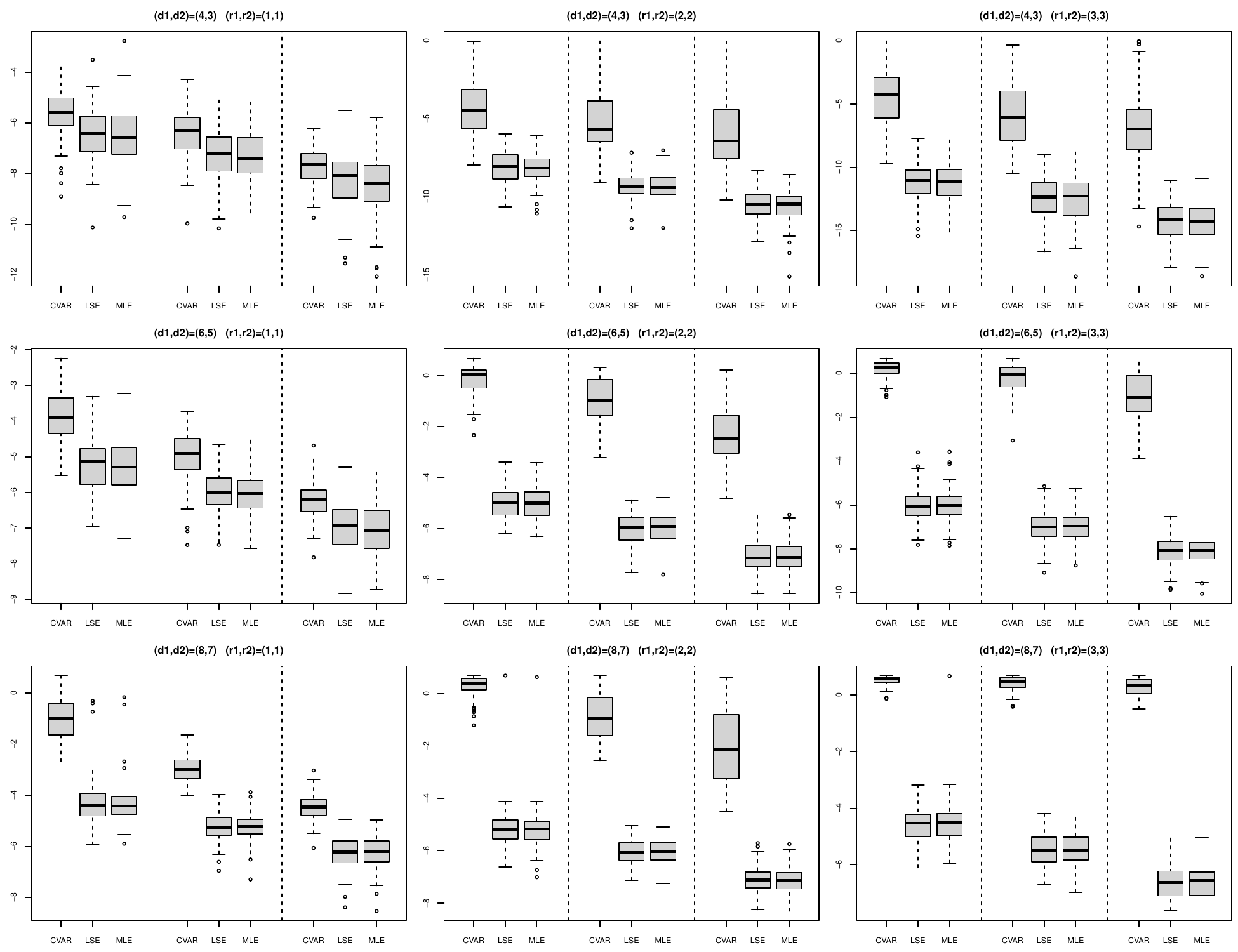}
    \caption{Comparison of CVAR, LSE and MLE. True model is the model~\ref{model form 2}, $k=1$, with the constant term, under setting~I. The three panels in each figure correspond to sample sizes 400, 600 and 1000 respectively.}
    \label{Figure:CMARSVD2}
\end{figure}

\begin{figure}[!ht]
    \centering
    \includegraphics[width = \textwidth]{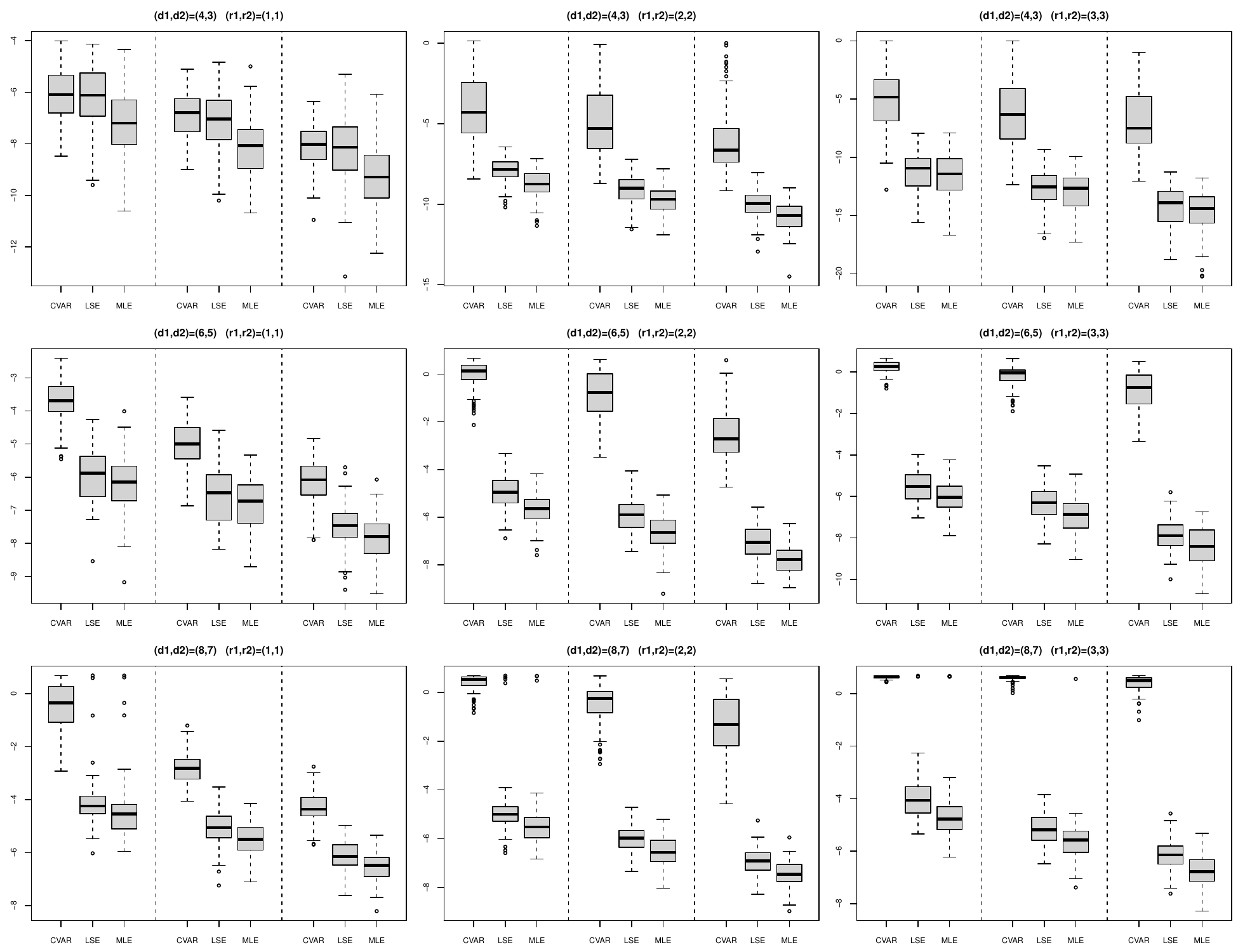}
    \caption{Comparison of CVAR, LSE and MLE. True model is the model~\ref{model form 2}, $k=1$, with the constant term, under setting~II. The three panels in each figure correspond to sample sizes 400, 600 and 1000 respectively.}
    \label{Figure:CMARMLE2}
\end{figure}

It is evident from both figures that the advantage of LSE and MLE over CVAR increases with higher dimensions. Furthermore, for fixed dimensions, this advantage remains significant as the ranks increase, since the CVAR does not consider the Kronecker structure of the cointegrating vectors. Additionally, from Figure~\ref{Figure:CMARSVD1}, under Setting~I of the error covariance matrix, MLE and LSE performs very similarly, even though the covariance matrix $\Sigma_e$ does not have the form~(\ref{sigma_e}). While Figure~\ref{Figure:CMARMLE1} explicitly depicts the advantages of MLE over LSE under Setting~II, which is assumed for MLE. The LSE and MLE also both performs much better than CVAR in the model including the constant term, as seen in Figure~\ref{Figure:CMARSVD2} with Setting~I and Figure~\ref{Figure:CMARMLE2} with Setting~II.

\subsection{Pairs Trading}\label{SEC:trading}

We apply our model to a well-known investment strategy on Wall Street, referred to as pairs trading, i.e., relative-value arbitrage strategies involving two or more securities. Pairs trading is a widely used trading strategy in the investment industry, popular among hedge funds and investment banks. The references provided are commonly cited works in the literature on pairs trading. \cite{vidyamurthy2004pairs} provides an introduction to pairs trading and related strategies, while \cite{gatev2006pairs} study the performance of pairs trading on US equities over 40 years, and \cite{krauss2017statistical} reviews the growing literature on pairs trading frameworks.

This strategy involves a two-step process. First, pairs of assets whose prices have historically moved together are identified during a formation period. Second, a trading strategy is designed to short the winner and buy the loser if the prices diverge and the spread widens. The key step is how to identify and formulate the pairs. Various methods have been proposed in the literature, including minimum distance between normalized historical prices \citep{gatev2006pairs}, and the cointegration analysis based on vector error-correction models \citep{vidyamurthy2004pairs, zhao2018mean}, or even simply by the intuition of experienced traders. Among all these methods, cointegration is a very interesting property that can be exploited in finance for trading. Idea is that while it may be difficult to predict individual stocks, it may be easier to predict relative behavior of stocks.

In this section, we focus on using the cointegrating vectors estimated by the proposed CMAR model for pairs formation. For example, we use daily data of Fama-French portfolios from July 2021 to Dec 2022 to formulate the pairs in the first 12 months, and subsequently implement a rolling trading strategy for the following six months. The data is publicly available at the data library maintained by Prof. Kennth R. French. Specifically, since the Fama-French portfolios are grouped into a $5 \times 5$ matrix-valued time series based on the Market Ratio (five levels from low to high) and Operating Profitability (five levels from low to high), we estimate the cointegrating vectors $\hat\bbeta=\hat\bbeta_2 \otimes \hat\bbeta_1$ using the CMAR model (\ref{model form 2}) with $k=1$ and including the constant term. Firstly, we employ the maximum likelihood estimation with $r_1=r_2=1$. Second, let $\hat\bbeta_{+}$ and $\hat\bbeta_{-}$ denote the positive and negative parts of $\hat\bbeta$, respectively. Then, the portfolios formulated by $\hat\bbeta' \vec(\bX_t) = \hat\bbeta'_{+} \bX_t - \hat\bbeta'_{-} \bX_t$ can be viewed as the spread between the two sub-portfolios, $\hat\bbeta'_{+}\bX_t$ and $\hat\bbeta'_{-}\bX_t$, where the negative sign corresponding to $\hat\bbeta'_{-} \bX_t$ is being shorted.

Since the cointegration model assumes that $\bbeta' \vec(\bX_t)$ is stationary, it can be viewed as an equilibrium relationship which will revert to its historical mean repeatedly. We design the trading strategy as follows: suppose the first day of trading is $t_0$, and denote $\hat\mu$, $\hat\sigma$ as the historical mean and standard deviation of $\{\hat\bbeta' \vec(\bX_{t})\}$ in the pairs formation period (the 12 months prior to $t_0$). Then we would have three situations on day one of trading $t_0$. (1) If $\hat\bbeta' \vec(\bX_{t_0}) \ge \hat\mu + s\hat\sigma$, we opened the position by longing $\hat\bbeta'_{-} \bX_{t_0}$ and shorting $\hat\bbeta'_{+} \bX_{t_0}$, then wait until the day $t_1$ that $\hat\bbeta' \vec(\bX_{t_1}) \le \hat\mu - s\hat\sigma$, we closed the position by selling all $\hat\bbeta'_{-} \bX_{t_0}$ and buying back all $\hat\bbeta'_{+} \bX_{t_0}$, where $s$ is the selected thresh-hold. (2) If on day one of trading, $t_0$, $\hat\bbeta' \vec(\bX_{t_0}) \le \mu - s\sigma$ the actions and positions would be reversed (long $\hat\bbeta'_{+} \bX_{t_0}$ and short $\hat\bbeta'_{-} \bX_{t_0}$ then close the position when $\hat\bbeta' \vec(\bX_{t_1}) \ge \hat\mu + s\hat\sigma$). (3) If neither of the above two cases occurs, we wait until one of these cases happens and denote that day as $t_0$, and then take actions accordingly.

In particular, we used a rolling strategy, such that each time when the position is closed at $t_1$, the parameters $\hat\bbeta$, $\hat\mu$, and $\hat\sigma$ are re-estimated using the data from the 12 months prior $t_1$. The money is then reinvested by opening long positions first, and $t_1$ becomes the new $t_0$, and the above trading rules are adopted again. The returns are compounded over time. We repeated this rolling strategy until the last day of the trading period $T$. On the last day $T$, the position was closed, regardless of the price.

Figure~\ref{pairsTrading2022} illustrates the movement of the spread $\hat\bbeta' \vec(\bX_t)$ from January 3, 2022 to December 30, 2022. The last position was opened on December 19, 2022 and closed on December 30, 2022. Although the spread did not touch the upper threshold on the closing day, it still increased, resulting in a positive cash flow for the final trade.

\begin{figure}[!ht]
    \centering
    \includegraphics[width = 0.9\textwidth]{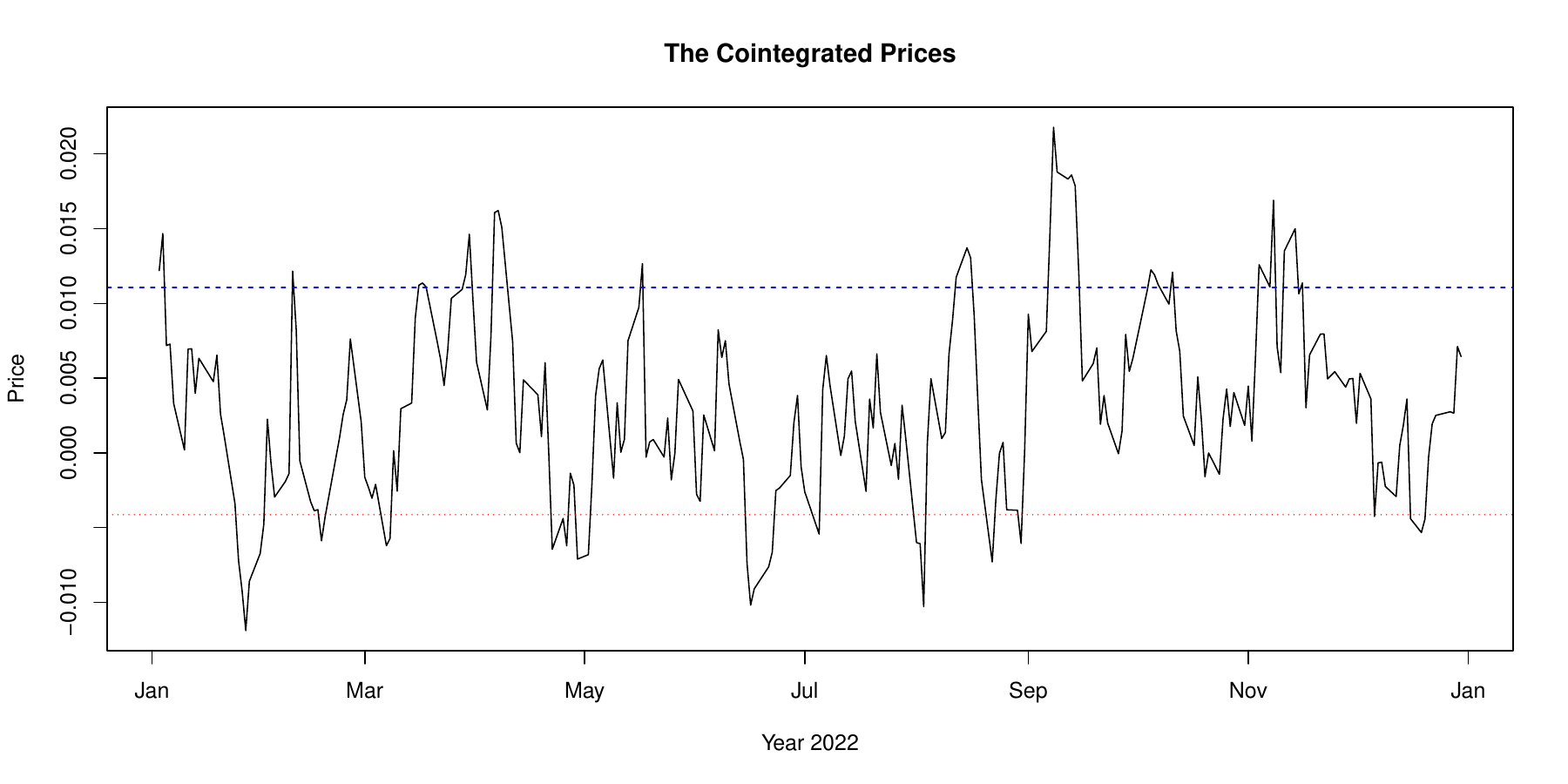}
    \caption{The cointegrated series $\hat{\bbeta}' \vec(\bX_t)$ in 2022, where $\hat{\bbeta}$ is estimated by the data from 2022/01/03 to 2022/12/30, the upper dashed line is $\hat\mu + s \hat{\sigma}$, the lower dashed line is $\hat{\mu} - s \hat{\sigma}$ with $s=1$.}
    \label{pairsTrading2022}
\end{figure}

The top row of Table~\ref{tradingReturns} displays the cumulative returns of our strategy during the trading period from July 1, 2022 to December 30, 2022. To provide a benchmark for comparison, we also tested the trading strategy proposed in \cite{gatev2006pairs}, in which equal value positions are opened, with one dollar long $\hat\bbeta_{-}$ and one dollar short $\hat\bbeta_{+}$ each time. Our strategy outperforms both the market and the benchmark for the year 2022. It is effective not only during bear markets but also when the market is relatively stable, as seen in the first half of 2018 and 2015. Furthermore, if we extend our trading period to the entire year, our strategy can beat the market in 2020.


\begin{table}[!ht]
\centering
\begin{tabular}{@{}lllll@{}}
\toprule
Trading Period   & $s=1.5$ & $s=1$ & $s=0.8$ &  S\&P \\ \midrule
\multirow{2}{*}{2022 Jul-Dec} & $3.041\%$  & $5.445\%$ & $0.897\%$ & $-7.477\%$ \\
                              & $2.421\%$  & $3.926\%$ & $0.568\%$ &  \\ \cmidrule(l){2-5}
\multirow{2}{*}{2022 Jan-Jun} & $1.913\%$  & $2.766\%$ & $2.806\%$ & $-20.869\%$ \\
                              & $1.519\%$  & $0.028\%$ & $0\%$ &  \\ \cmidrule(l){2-5}
\multirow{2}{*}{2018 Jan-Jun} & $1.815\%$  & $1.901\%$ & $3.532\%$ & $0.791\%$ \\
                              & $-0.431\%$  & $-0.519\%$ & $-1.698\%$ &  \\ \cmidrule(l){2-5}
\multirow{2}{*}{2015 Jan-Jun} & $0.721\%$  & $2.193\%$ & $1.474\%$ & $0.437\%$ \\
                              & $-0.706\%$  & $-0.308\%$ & $0.147\%$ &  \\ \cmidrule(l){2-5}
\multirow{2}{*}{2020} & $7.970\%$  & $16.730\%$ & $16.719\%$ & $16.26\%$ \\
                              & $-2.352\%$  & $15.726\%$ & $15.056\%$ &  \\ 
\bottomrule
\end{tabular}
\caption{The cumulative returns for different trading periods with different threshold values ($s=0.8, 1, 1.5$). There are two rows where the first row showing the returns of our strategy and the second for the equal value strategy.}
\label{tradingReturns}
\end{table}

The equal value strategy does not perform well since it uses pairs formulated by cointegration but does not trade them in proportion to their share ratio. While our strategy performs well for almost all threshold values, we still need to explore how to select the threshold in a consistent way. Moreover, since pairs trading can be viewed as a hedging method, making it difficult to beat the market in a strong bull market like 2021, however, it is particularly effective in the bear markets like 2022.

\section{Conclusions}\label{SEC:conclutions}

We introduce the cointegrated matrix autoregressive model (CMAR), which relies on an autoregressive
term involving bilinear cointegrating vectors in the error-correction form, with rank deficiency of the coefficient matrices. Comparing with the traditional vector error-correction model (CVAR) \citep{engle1987co, johansen1995likelihood}, the CMAR model considers the natural matrix structure of the data and hence can better reveals the underlying cointegration relations among the variables, leads to more efficient estimation. On the other hand, we use the MAR estimates as the warm-start initial values for the estimation of the CMAR model. Both LSE and MLE are studied, where the latter is considered under an additional assumption that the covariance tensor of the error matrix is separable. Our numerical analysis suggests that even if the separability assumption on the covariance tensor does not hold, MLE still has reasonable and almost equally good performance, comparing with LSE. On the other hand, MLE can perform much better when that assumption does stand. Therefore, we would recommend the use of MLE in practice.

There are a number of directions to extend the study of the cointegrated matrix autoregressive model.
Firstly, the bilinear form of the cointegrating vectors in the proposed model renders traditional testing methods inadequate, thus motivating the exploration of novel testing procedures specific to the matrix-form model. Moreover, studying processes with hybrid cointegrated orders would be more practical and interesting. However, testing and estimation for such models can be more challenging. More importantly, the asymptotic analysis has been carried out for the fixed dimensional case in
the current paper. It is interesting and important to study the model under the high dimensional
paradigm. Also, the model can be extended for tensor time series as well.

\newpage
\bibliography{mybibfile}
\newpage

\begin{appendices}\label{Appendix}

\section{Additional Theorems}\label{Appendix:theorems}

\subsection{Asymptotic for the constant term}

The asymptotic properties of the constant term $\hat\bd = \vec(\hat{\bD})$ can be obtained directly from \cite{johansen1995likelihood}. For convenience, we reproduce the relevant results below.
Define \[\bY_{t}(\bbeta)' = ((\bbeta' \vec'{(\bX_{t-1})}), \Delta\vec'{(\bX_{t-1})}, \ldots, \Delta\vec'{(\bX_{t-k})}),\]
We denote $\Omega = \Var{(\bY_{t}(\bbeta))}$, which can be calculated by the known parameters, since both $\bbeta' \vec{(\bX_{t-1})}$ and $\Delta\bX_{t}$ are stationary. And we have
\[\mathbb{E}(\bY_{t}(\bbeta))' = \vec'{(\bD)} \{(\bC' \Gamma' - \bI)\bar{\balpha}, \bC', \ldots, \bC'\} :=  \vec'{(\bD)} \xi'\]
where $\bC = \bbeta_{\bot} (\balpha_{\bot}' \Gamma \bbeta_{\bot})^{-1} \balpha_{\bot}'$, $\Gamma = \bI - \sum_{i=1}^{k} \bB_{i2} \otimes \bB_{i1}$. Let $\h{\tau} = \bC \bd$ and choose $\h{\gamma}$ orthogonal to $\bbeta$ and $\h{\tau}$, which $(\h{\tau}, \h{\gamma}, \bbeta)$ span the whole space. Define
\[
G_0(u) := \begin{pmatrix}
     \bar{\h{\gamma}}' \bC \bw(u) \\
     u
\end{pmatrix}, \
G(u) := G_0(u) - \int_0^1 G_0(u) du =
\begin{pmatrix}
     \bar{\h{\gamma}}' \bC \left(\bw(u) - \int_0^1 \bw(u)\right) \\
     u - 1/2
\end{pmatrix}.
\]

\begin{theorem}
The asymptotic distribution of the constant term $\hat\bd = \vec(\hat{\bD})$ is found from
\begin{align}\label{CLTD}
\begin{split}
    \sqrt{T}(\hat{\bd} - \bd) \to& N(\bzero, \Sigma \bd' \xi' \Omega^{-1} \xi \bd) \\
    +& \bw(1) - \balpha \int_{0}^{1} (dV_{\alpha}) G' \left[\int_{0}^{1}GG' du\right]^{-1} \int_0^1 G_0(u) du,
\end{split}
\end{align}
where $V_{\alpha} = \left(\balpha' \Sigma^{-1} \balpha \right)^{-1} \balpha' \Sigma^{-1} \bw(u)$.
\end{theorem}
\begin{proof}
See Chapter~13 of \cite{johansen1995likelihood}.
\end{proof}

\subsection{Asymptotic for the model without the constant term}
We only highlight some major differences and show some results for LSE as examples. In Section~\ref{SEC:Theorems}, we investigated the process in three different directions, where $(\bbeta_1, \h{\tau}_1, \h{\gamma}_1)$ has a full rank of $d_1$. However, in the following, we only consider two directions $(\bbeta_1, \bbeta_{1\bot})$. Consider LSE estimator for example.
\begin{equation}
    \hat{\hA}_1, \hat{\hA}_2 = \argmin_{\hA_1, \hA_2} \sum_t\|\Delta \hX_t - \hA_1\hX_{t-1}\hA_2' \|^2_F, \quad \hat{\hA}_i = \hat{\balpha}_i \hat{\bbeta}_i'.
\end{equation}
Since we can view the model as 
\begin{equation}\label{LSE}
    \Delta \bX_t = \balpha_1 \boxed{\bbeta_1' \bX_{t-1} \bbeta_2} \balpha_2' + \bE_t,
\end{equation}
where $\bbeta_1\bX_{t-1} \bbeta_2'$ is $r_1 \times r_2$ stationary matrix-valued time series. We can get the CLT for $\hat{\balpha}_1, \hat{\balpha}_2$. Define
\begin{align*}
\boldsymbol{W}_{t-1} = \begin{pmatrix} 
\bbeta_1' \bX_{t-1} \bbeta_2 \balpha_2' \otimes \bI_{d_1}\\
\bI_{d_2} \otimes \bbeta_2' \bX_{t-1}' \bbeta_1 \balpha_1'\\
\end{pmatrix},
\end{align*}
and $\bH = \mathbb{E} \left( \bW_t \bW_t'\right) + \h{\gamma} \h{\gamma}'$.
\begin{theorem}[\textit{CLT for $\balpha$}]\label{tildealpha}
The asymptotic distribution of $\tilde{\balpha}_1 := \hat{\balpha}_1 \hat{\bbeta}_1' \bbeta_1$ is given by
$$T^{\frac{1}{2}} (\tilde{\balpha}_1 - \balpha_1) \to N_{d_1 \times r_1} \left(0, \Xi_1 \right),$$
where $\Xi_1 = \bH^{-1}\mathbb{E}(\bW_t \Sigma \bW_t')\bH^{-1}$.
\end{theorem}

The process $\bX_t$ of the model without the constant term does not have the linear trend and is given as a mixture of a random walk and a stationary process:
\begin{align}
    \vec{(\bX_t)} = \bC \sum_{i=1}^{t} \bepsilon_i + C(L)(\bepsilon_t) + \vec{(\bX_0)},
\end{align}
where $\bepsilon_i = \vec{(\bE_i)}$, $\bd = \vec{(\bD)}$, $\bC = \bbeta_{\bot} (\balpha_{\bot}' \bbeta_{\bot})^{-1} \balpha_{\bot}'$, and $C(L)(\bepsilon_t)$ is the stationary part. Thus, we only consider two directions $(\bbeta_1, \bbeta_{1\bot})$.

\begin{lemma}
\begin{align*}
    T^{-\frac{1}{2}} \bar{\bbeta}_{1\bot}' \bX_{[Tu]} \bA_2' \overset{w}{\longrightarrow} \bar{\bbeta}_{1\bot}' \left[ \vec^{-1}(\bC\bw(u)) \right] \bA_2',\\
        \left[T^{-\frac{1}{2}}  \bA_1 \bX_{[Tu]} \bar{\bbeta}_{2\bot}\right]' \overset{w}{\longrightarrow} \left[\bA_1 \left[ \vec^{-1}( \bC \bw(u)) \right] \bar{\bbeta}_{2\bot}\right]').
\end{align*}
\end{lemma}

\begin{theorem}
Let $\tilde{\balpha}_i = \hat{\balpha}_i \hat{\bbeta}_i' \bbeta_i$, $\tilde{\bbeta}_i = \hat{\bbeta}_i(\bbeta_i' \hat{\bbeta}_i)^{-1}$. Then we have
\begin{align}
    T\begin{pmatrix}
        \vec{\left((\tilde{\bbeta}_1 - \bbeta_1)'\bbeta_{1\bot}\right)} \\
        \vec{\left(\bbeta_{2\bot}'(\tilde{\bbeta}_2 - \bbeta_2)\right)}
    \end{pmatrix} \overset{w}{\longrightarrow}
    \left(\int_0^1 \bG_3\bG_3' \mathrm{d}u \right)^{-1} \int_0^1 \bG_3 (\mathrm{d}\bw).
\end{align}
where $\bG_3$ is defined as
\[
\bG_3(u) :=
\begin{pmatrix}
     \bar{\bbeta}_{1\bot}' \left[ \vec^{-1}(\bC\bw(u)) \right] \bA_2' \otimes \balpha_1' \\
     \balpha_2 \otimes \left[\bA_1 \left[ \vec^{-1}( \bC \bw(u)) \right] \bar{\bbeta}_{2\bot}\right]'
\end{pmatrix}.
\]
\end{theorem}
We omit the details for MLE since the theorems can be derived similarly with respect to the above two directions $(\bbeta_1, \bbeta_{1\bot})$.

\section{Proof of the Theorems}\label{Appendix:proofs}

We collect the proofs of Lemma~\ref{G}, Theorem~\ref{Theorem:LSEtheta}, Theorem~\ref{Theorem:LSEtildebeta}, Corollary~\ref{Corollary:LSEtildebeta}, Theorem~\ref{Theorem:LSEhatbeta}, Theorem~\ref{Theorem:LSEbb} in this section. Proof of Theorem~\ref{Theorem:represent} can be found in Chapter~4 \cite{johansen1995likelihood}. For the following proofs, unless otherwise specified, we omit the subscript $\hat\bA^{\text{ls}}$ or $\hat\bA^{\text{ml}}$ for brevity and simply use $\hat\bA$ to denote the estimator used in each particular theorem.

\subsection{Proof of the Lemma~\ref{G}}
\begin{proof}
Since the vectorized CMAR process has the representation (\ref{vecProcess}) by Theorem~\ref{Theorem:represent}, we see that $\vec(\bX_t)$ is decomposed into a random rank, a linear trend, and a stationary process. Taking inverse vectorization of both sides of (\ref{vecProcess}), such that $\vec^{-1}$ makes the vector $\vec(\bX_t)$ back into a matrix $\bX_t$, the weak convergence of $T^{-\frac{1}{2}} \bar{\h{\gamma}}'_1 \bX_{[Tu]}$ follows from Theorem~\ref{Theorem:represent}. Since the random walk part $T^{-\frac{1}{2}} \bar{\h{\gamma}}'_1 [\vec^{-1}(\bC \sum_i^{t} \bepsilon_i)]\bbeta_2$ gives $\bar{\h{\gamma}}'_1 \bW(u) \bbeta_2$ in the limit; and the trend part vanishes because $\bar{\h{\gamma}}'_1 \vec^{-1}(\bC \bd) \bbeta_2 =\bzero$ by the orthogonal construction of $\h{\gamma}_1$ and $\h{\tau}_1 = \col(\vec^{-1}(\bC \bd) \bbeta_2)$; and the stationary part also vanishes by the factor $T^{-\frac{1}{2}}$.

The weak convergence of $T^{-1} \bar{\h{\gamma}}'_1 \bX_{[Tu]}$ can be also derived from (\ref{vecProcess}). The random walk term  and the stationary term tend to zero by the factor $T^{-1}$, the linear term converges to $u$.

And note that the weak convergence of the average $T^{-\frac{1}{2}} \bB'_{1T} \bar{\bX}_t = T^{-1} \sum_{t=1}^{T} T^{-\frac{1}{2}} \bB'_{1T}\bX_{t-1}$ follows from the continuous mapping theorem, where the function $x \to \int_0^1 x(u) du$ is continuous.

\end{proof}

\subsection{Proof of Theorem~\ref{Theorem:LSEtheta}}
To prove Theorem~\ref{Theorem:LSEtheta}, we first state and prove the following lemma. Consider the general form
\begin{equation}\label{Proof:modelForm1}
    \Delta \bX_t = \bA_1 \bX_{t-1} \bA_2' + \sum_{i=1}^{k} \bB_{i1} \Delta \bX_{t-i} \bB_{i2}' +  \bD  + \bE_t, \ t=1,\ldots,T,
\end{equation}
where $\bD$ is a constant $d_1 \times d_2$ matrix, and $\bA_1 = \balpha_1 \bbeta_1'$, $\bA_2 = \balpha_2 \bbeta_2'$. Denote $\bd = \vec(\bD)$, $\Phi = \bA_2 \otimes \bA_1$, $\bB_i = \bB_{i2} \otimes \bB_{i1}$.
We introduce the notation for the stacked vectors and the corresponding parameters,
\[\bZ_t' = \left[ (\bU_T \vec{(\bX_{t-1})})',\ \Delta \vec{(\bX_{t-1})}', \ldots, \Delta \vec{(\bX_{t-k})}' \right], \bC' = \left[ \Phi\bB_T, \bB_1, \ldots, \bB_{k} \right],\]
where $\bU_T := (\bbeta, T^{-\frac{1}{2}} \bbeta_{\bot})' \in \mathbb{R}^{d_1 d_2 \times d_1 d_2}$ which is full rank so invertible.
$\bB_T:=(\bar{\bbeta}, T^{\frac{1}{2}} \bar{\bbeta}_{\bot})$ where $\bar{\bbeta} = \bbeta (\bbeta' \bbeta)^{-1}$. Note that $\bB_T = \bU_T^{-1}$. 

Using the above notations, we can rewrite the model (\ref{Proof:modelForm1}) as follows:
\begin{equation}\label{Proof:modelForm1vec}
    \Delta \vec{(\bX_{t})} = \bC' \bZ_t + \bd + \vec{(\bE_t)}, \ t=1,\ldots,T.
\end{equation}
The following lemma states the $\hat\bC$ is asymptotic normal with the convergence rate $O_p(1/\sqrt{T})$. In addition, the projection matrix $\mathbb{P}_{\beta} = \bbeta_i \bbeta_i'$ is super-consistent with the convergence rate $O_p(1/T)$.
\begin{lemma}\label{LSE_rate}
For the LSE estimators,
\begin{enumerate}
    \item[(a)] $(\hat{\Phi} - \Phi)' \bB_T = O_p(\frac{1}{\sqrt{T}})$;
    \item[(b)] $(\hat{\bB} - \bB) = O_p(\frac{1}{\sqrt{T}})$;
    \item[(c)] $(\hat{\bD} - \bD) = O_p(\frac{1}{\sqrt{T}})$;
    \item[(d)] $(\hat{\balpha}_i - \balpha_i) = O_p(\frac{1}{\sqrt{T}}), i=1,2$;
\end{enumerate}
where $\bB_T:=(\bar{\bbeta}, T^{\frac{1}{2}} \bar{\bbeta}_{\bot})$, $\bar{\bbeta} = \bbeta (\bbeta' \bbeta)^{-1}$.
Moreover,
\begin{enumerate}
    \item[(e)] $\hat{\bbeta}_i\hat{\bbeta}'_i - \bbeta_i \bbeta'_i = O_p(\frac{1}{T}), i=1,2$.
\end{enumerate}
\end{lemma}
\begin{proof}
Recall that $\bbeta = \bbeta_2 \otimes \bbeta_1 \in \mathbb{R}^{d_1d_2 \times r_1 r_2}$, $\bbeta_{\bot}$ orthogonal to $\bbeta$ such that $\bbeta_{\bot}' \bbeta = \bzero$, $\bbeta_{\bot} \in \mathbb{R}^{d_1d_2 \times (d_1 d_2 - r_1 r_2)}$. To prove (a) to (d), it is sufficient to show for any sequence ${c_T}$ such that $c_T \to \infty$,
\begin{equation}\label{lemma2.0}
P\left[\inf_{\sqrt{T}\|\bar{\bC} - \bC\|_F \ge c_T} \sum_{t=1}^{T} \|\bY_t - \bar{\bC}' \bZ_t\|^2 \le \sum_{t=2}^{T} \|\vec{(\bE_t)}\|^2 \right] \to 0.
\end{equation}
Without loss of generality, to simplify notations, consider the model without term $\bB$ and $\bD$,
\begin{equation}\label{Proof:modelForm2}
    \Delta \bX_t = \bA_1 \bX_{t-1} \bA_2' + \bE_t, \ t=1,\ldots,T.
\end{equation}
The proof can be extended to the model (\ref{Proof:modelForm1}) under same idea. Then, to prove (a), (b) and (d), it is sufficient to show for any sequence ${c_T}$ such that $c_T \to \infty$,
\begin{equation}\label{lemma2.1}
P\left[\inf_{\sqrt{T}\|(\bar{\Phi} - \Phi) \h{B}_T\|_F \ge c_T} \sum_{t=1}^{T} \|\vec{(\bX_t) - \bar{\Phi}} \vec{(\bX_{t-1})}\|^2 \le \sum_{t=2}^{T} \|\vec{(\bE_t)}\|^2 \right] \to 0.
\end{equation}
Now we write
\begin{align}\label{bigMinusSmall}
\begin{split}
    &\sum_{t=2}^{T} \|\vec{(\bX_t)} - \bar{\Phi}\vec{(\bX_{t-1})}\|^2 - \sum_{t=2}^{T} \|\vec{(\bE_t)}\|^2 \\
    =&-2\sum_{t=2}^{T} \tr{[(\bar{\Phi} - \Phi) \vec{(\bX_{t-1})} \vec{(\bE_t)'}] + \sum_{t=2}^{T} \tr{[(\bar{\Phi} - \Phi) \vec{(\bX_{t-1})} \vec{(\bX_{t-1})'} (\bar{\Phi} - \Phi)' ]}}.
\end{split}
\end{align}
Denote $\bS_T := \frac{1}{T} \sum_{t=1}^{T} \bU_T \vec{(\bX_{t-1})} \vec{(\bX_{t-1})}' \bU_T'$. Note that
$$\sum_{t=2}^{T} \tr{[(\bar{\Phi} - \Phi) \vec{(\bX_{t-1})} \vec{(\bX_{t-1})'} (\bar{\Phi} - \Phi)']} 
= T \tr{[(\bar{\Phi} - \Phi) \bB_T \bS_T \bB_T' (\bar{\Phi} - \Phi)']}.$$
Note $\sum_{t=1}^{T} \bU_T \vec{(\bX_{t-1})} = \sum_{t=1}^{T} \begin{pmatrix}
     \bbeta' \vec{(\bX_{t-1})} \\
     T^{-\frac{1}{2}} \bbeta'_{\bot} \vec{(\bX_{t-1})}
\end{pmatrix}$. By Representation Theorem~\ref{Theorem:represent}, we know that $\bbeta' \vec{(\bX_{t-1})}$ is I(0) and $ \bbeta'_{\bot} \vec{(\bX_{t-1})}$ is I(1). Assuming the fourth moment of $\bE_t$ exists, by calculating the variance, we know that
\begin{equation}
    \bS_T := \frac{1}{T} \sum_{t=1}^{T} \bU_T \vec{(\bX_{t-1})} \vec{(\bX_{t-1})}' \bU_T' \overset{w}{\longrightarrow}\Gamma_0
\end{equation}
where $\Gamma_0$ is a symmetric positive definite matrix. Thus, exists sequence $\{d_T\}$ such that $d_T \to 0$ and $d_Tc_T \to \infty$,
$$P\left(\lambda_{\min} (\Gamma_0) > d_T \right) \to 1 \ \text{as} \ T \to \infty.$$
Since $\bS_T \overset{w}{\longrightarrow}\Gamma_0$, we have $\lambda_{\min}(\bS_T) \overset{w}{\longrightarrow} \lambda_{\min}(\Gamma_0)$. By Portmanteau lemma,
$$P\left(\lambda_{\min} (\bS_T) > d_T \right) \to P\left(\lambda_{\min} (\Gamma_0) > d_T \right) \to 1 \ \text{as} \ T \to \infty.$$
So,
\begin{equation}
    P\left(\tr{[(\bar{\Phi} - \Phi) \bB_T \bS_T \bB_T' (\bar{\Phi} - \Phi)']} > T \|(\bar{\Phi} - \Phi)\h{B}_T\|^2_F d_T \right) \to 1  \ \text{as} \ T \to \infty.
\end{equation}
On the boundary set $\sqrt{T}\|(\bar{\Phi} - \Phi) \h{B}_T\|_F = c_T$,
\begin{equation}\label{bigRate}
    T \tr{[(\bar{\Phi} - \Phi) \bB_T \bS_T \bB_T' (\bar{\Phi} - \Phi)']} > O_p(T c_T^2 d_T) 
\end{equation}
On the other hand, by calculating the variance, we know that
\begin{equation}\label{smallRate}
    \sum_{t=2}^{T} \tr{[(\bar{\Phi} - \Phi) \vec{(\bX_{t-1})} \vec{(\bE_t)'}]} = O_p(Tc_T)
\end{equation}
Combining (\ref{bigMinusSmall}), (\ref{bigRate}) and (\ref{smallRate}), we have
\begin{equation}\label{lemma2.2}
    P\left[\inf_{\sqrt{T}\|(\bar{\Phi} - \Phi) \h{B}_T\|_F = c_T} \sum_{t=1}^{T} \|\vec{(\bX_t) - \bar{\Phi}} \vec{(\bX_{t-1})}\|^2 \le \sum_{t=2}^{T} \|\vec{(\bE_t)}\|^2 \right] \to 0.
\end{equation}
Since $\sum_{t=2}^{T} \|\vec{(\bX_t)} - \bar{\Phi}\vec{(\bX_{t-1}}\|^2$ is a convex function of $\bar{\Phi}$, so (\ref{lemma2.1}) is implied by (\ref{lemma2.2}). 

Therefore, (a) and (b) is implied by (\ref{lemma2.1}). To prove (c), consider same idea for (\ref{lemma2.0}). And (d) is implied by (a).

To prove (e), by (a) we know that 
\begin{equation}\label{betabotRate}
    (\hat{\bbeta} - \bbeta)' \bbeta_{\bot} = O_p(\frac{1}{T}).
\end{equation}
And since $\hat{\bbeta}' \hat{\bbeta} = \bbeta' \bbeta = \bI_r$, we have 
\begin{equation}\label{bb'}
    (\hat{\bbeta} - \bbeta)' \bbeta + \bbeta' (\hat{\bbeta} - \bbeta) = O_p(\frac{1}{T}).
\end{equation}
Note that
\begin{equation}\label{b'b}
    \hat{\bbeta} \hat{\bbeta}' - \bbeta \bbeta'  = (\hat{\bbeta} - \bbeta) \bbeta' + \bbeta (\hat{\bbeta} - \bbeta)' + O_p(\frac{1}{T}).
\end{equation}
By (\ref{betabotRate}) and (\ref{b'b}) we know that
\begin{equation}\label{lemma2.3}
    (\hat{\bbeta} \hat{\bbeta}' - \bbeta \bbeta') \bbeta_{\bot} = O_p(\frac{1}{T}).
\end{equation}
So we only need to show $(\hat{\bbeta} \hat{\bbeta}' - \bbeta \bbeta') \bbeta = O_p(\frac{1}{T})$. Since
\begin{align}\label{lemma2.4}
\begin{split}
    (\hat{\bbeta} \hat{\bbeta}' - \bbeta \bbeta') \bbeta  = &(\hat{\bbeta} - \bbeta) \bbeta'\bbeta + \bbeta (\hat{\bbeta} - \bbeta)'\bbeta + O_p(\frac{1}{T}) \\
    =& \hat{\bbeta} - \bbeta - \bbeta \bbeta' (\hat{\bbeta} - \bbeta) + O_p(\frac{1}{T}) \\
    =& (\bI - \bbeta \bbeta')  (\hat{\bbeta} - \bbeta) + O_p(\frac{1}{T}) \\
    =& \bbeta_{\bot} \bbeta'_{\bot} (\hat{\bbeta} - \bbeta) + O_p(\frac{1}{T}) = O_P(\frac{1}{T}).
\end{split}
\end{align}
The (e) is followed by (\ref{lemma2.3}) and (\ref{lemma2.4}).
\end{proof}
Now we are ready to give the proof of Theorem~\ref{Theorem:LSEtheta}.
\begin{proof}
By the gradient condition of (\ref{minlse}) for $\balpha_1$, $\balpha_2$ and $\bB_{i1}$, $\bB_{i2}$, we have
\begin{align}\label{grad:lse}
\begin{split}
\hat\bR_t \hat{\balpha}_2 \hat{\bbeta}'_2 \bX'_{t-1} \hat{\bbeta}_1 &= \bzero, \\
\hat\bR_t \hat{\balpha}_1 \hat{\bbeta}'_1 \bX'_{t-1} \hat{\bbeta}_2 &= \bzero, \\
\hat\bR_t \hat{\bB}_{i2} \Delta \bX'_{t-1} &= \bzero, \\
\hat\bR_t \hat{\bB}_{i1} \Delta \bX_{t-1} &= \bzero.
\end{split}
\end{align}
where $\hat\bR_t = \Delta \bX_t  - \hat{\bA}_1 \bX_{t-1} \bA_2' - \sum_{i=1}^{k} \hat\bB_{i1} \Delta \bX_{t-i} \bB_{i2}' -  \hat\bD$, and we plug in the formula (\ref{Proof:modelForm1}) for $\Delta \bX_t$.
In particular, for $\hat\bD$, we have
\begin{equation}\label{grad:hatD}
    \hat{\bD} = T^{-1}\sum_{t=1}^{T} \left( \Delta\bX_t -  \hat\bA_1  \bX_{t-1} \hat{\bA}'_2 - \sum_{i=1}^{k} \hat{\bB}_{i1} \Delta\bX_{t-1} \hat{\bB}'_{i2}\right).
\end{equation}
Take the (\ref{grad:hatD}) into gradient conditions (\ref{grad:lse}), and note that
\begin{align*}
\sum_{t=1}^{T} \bbeta_1' (\bX_{t-1} - \bar{\bX}_{t-1}) \bbeta_2 \balpha_2' \balpha_2 \bbeta_2' \bX_{t-1}' \bbeta_1 &= \sum_{t=1}^{T} \bbeta_1' (\bX_{t-1} - \bar{\bX}_{t-1}) \bbeta_2 \balpha_2' \balpha_2 \bbeta_2' (\bX_{t-1} - \bar{\bX}_{t-1})' \bbeta_1, \\
\sum_{t=1}^{T} (\bE_t - \bar{\bE}_t) \balpha_2 \bbeta_2' \bX_{t-1}' \bbeta_1 &= \sum_{t=1}^{T} (\bE_t - \bar{\bE}_t) \balpha_2 \bbeta_2' (\bX_{t-1} - \bar{\bX}_{t-1})' \bbeta_1.
\end{align*}
Then taking vectorization of (\ref{grad:hatD}), and after some algebra we have
\[
\bH_1 \vec(\hat{\h{\theta}} - \h{\theta}) = \frac{1}{T} \sum_t \bW_{t-1} \vec(\bE_t) + o_p(T^{-1/2}),
\]
where $\bH_1$ and $\bW_t$ are defined at the beginning of Section~\ref{SEC:Theorems} that
\begin{align*}
\boldsymbol{W}_{t-1} = \begin{pmatrix} 
\bbeta_1' (\bX_{t-1} - \bar{\bX}_{t-1}) \bbeta_2 \balpha_2' \otimes \bI_{d_1}\\
\bI_{d_2} \otimes \bbeta_2' (\bX_{t-1} - \bar{\bX}_{t-1})' \bbeta_1 \balpha_1'\\
(\Delta\bX_{t-i} - \Delta\bar{\bX}_{t-i})  \bB_{i2} \otimes \bI_{d_1} \\
\bI_{d_2} \otimes (\Delta\bX_{t-i} - \Delta\bar{\bX}_{t-i})' \bB_{i1} \\
i=1,\ldots,k 
\end{pmatrix},
\end{align*}
$\bH_1 = \mathbb{E} \left( \bW_t \bW_t'\right) + \sum_{i=1}^{k}\h{\delta}_i \h{\delta}_i'$, where $\h{\delta}_1 = (\vec{(\balpha_1)}', \h{0}')'$, and $\h{\delta}_i = (\h{0}_{d_1r_1+id_2r_2}', \vec{(\bB_{i1})}', \h{0}')'$ for $i=2,\ldots,k$.
By martingale central limit theorem \citep{hall2014martingale},
\[
\sum_t \bW_{t-1} \vec(\bE_t) \to N(\bzero, \mathbb{E}(\bW_t \Sigma \bW_t')).
\]
It follows that
\[\sqrt{T} \vec{(\hat{\h{\theta}} - \h{\theta})} \to N(\h{0}, \Xi_1),\]
where $\Xi_1 = \bH_1^{-1}\mathbb{E}(\bW_t \Sigma \bW_t')\bH_1^{-1}$.
\end{proof}
\subsection{Proof of Theorem~\ref{Theorem:LSEtildebeta}}
To prove Theorem~\ref{Theorem:LSEtildebeta}, we first state and prove the following lemma.
\begin{lemma}\label{GG}
Define
\begin{align*}
\boldsymbol{M}_{t-1} = \begin{pmatrix} 
\bB_{T1}' (\bX_{t-1} - \bar\bX_{t-1}) \bbeta_2 \balpha_2' \otimes \balpha_1'\\
\balpha_2' \otimes \bB_{T2}' (\bX_{t-1} - \bar\bX_{t-1})'\bbeta_1 \balpha_1'\\
\end{pmatrix} \in \mathbb{R}^{(r_1d_1+r_2d_2 -r_1^2-r_2^2 ) \times d_1d_2}.
\end{align*}
Then
\begin{align}
T^{-\frac{1}{2}} \boldsymbol{M}_{[Tu]} &\overset{w}{\longrightarrow} 
 \bG(u) 
 \label{M}\\
T^{-2} \sum_t^{T} \boldsymbol{M}_{t-1} \boldsymbol{M}_{t-1}' &\overset{w}{\longrightarrow}\displaystyle\int_{0}^{1} \bG \bG' \mathrm{d}u, \label{MM}\\
T^{-1} \sum_t^{T} \boldsymbol{M}_{t}  \vec(\bE_t)  &\overset{w}{\longrightarrow} \int_0^1 \bG (\mathrm{d}\bW).\label{ME}
\end{align}
\end{lemma}
\begin{proof}
The (\ref{M}) follows from the Lemma~\ref{G} and the definition of $\bG(u)$. The (\ref{MM}) follows from the (\ref{M}) together with the continues mapping theorem applied to the mapping $\bG \to \int_0^1 \bG(u)\bG'(u) du$.
\end{proof}

\begin{proof}[Proof of Theorem~\ref{Theorem:LSEtildebeta}]
Note $T \bU_{T1} = T (\h{\gamma}_1, T^{\frac{1}{2}} \h{\tau}_1)' \bB_{T1} \bU_{T1} = (T\h{\gamma}_1, T^{\frac{3}{2}} \h{\tau}_1)' (\tilde{\bbeta}_1 - \bbeta_1)$.
By the gradient condition we have
\begin{align*}
    \sum_t \balpha_1' \left(\balpha_1(\tilde\bbeta_1 - \bbeta_1)'\bX_{t-1}\bbeta_2\balpha_2' + \balpha_1\bbeta_1' (\tilde\bbeta_2 - \bbeta_2) \balpha_2'\right)
    \balpha_2 \bbeta_2' \bX_{t-1}' &= \balpha_1' \bE_t \balpha_2 \bbeta_2' \bX_{t-1} +o_p(1), \\
    \sum_t \bX_{t-1}' \bbeta_1 \balpha_1' \left(\balpha_1(\tilde\bbeta_1 - \bbeta_1)'\bX_{t-1}\bbeta_2\balpha_2' + \balpha_1\bbeta_1' (\tilde\bbeta_2 - \bbeta_2) \balpha_2'\right)
    \balpha_2 &= \bX_{t-1}' \bbeta_1 \balpha_1' \bE_t \balpha_2 +o_p(1).
\end{align*}
We can rewrite above equations in the matrix form,
\begin{align*}
    \frac{1}{T^2} \sum_t \left( \bM_{t-1} \bM_{t-1}' \right)
    T
    \begin{bmatrix}
        \vec(\bU_{T1}') \\
        \vec(\bU_{T2})
    \end{bmatrix}
    = \frac{1}{T} \sum_{t} \bM_{t-1} \vec(\bE_t) + o_p(1/T).
\end{align*}
The results then follows from Lemma~\ref{G} and Lemma~\ref{GG}.
\end{proof}

\begin{proof}[Proof of Corollary~\ref{Corollary:LSEtildebeta}]
    (\ref{Formula:tildebeta1}) and (\ref{Formula:tildebeta2}) are the marginal distribution of the first and second components in (\ref{Formula:LSEtildebeta}).
Since 
\[\int_{0}^{1} \bG\bG' \mathrm{d}u =
\begin{pmatrix}
     \int_{0}^{1}\bG_{1}\bG_{1}'\mathrm{d}u & \int_{0}^{1}\bG_{1}\bG_{2}'\mathrm{d}u \\
     \int_{0}^{1}\bG_{2}\bG_{1}'\mathrm{d}u & \int_{0}^{1}\bG_{2}\bG_{2}'\mathrm{d}u
\end{pmatrix}.\]
After some algebra we have 
\[\int_{0}^{1} \bG_{1.2}\bG_{1.2}' \mathrm{d}u =
\int_{0}^{1}\bG_1\bG_1'\mathrm{d}u - \int_{0}^{1}\bG_1\bG_2'\mathrm{d}u(\int_{0}^{1}\bG_2\bG_2'\mathrm{d}u)^{-1}\int_{0}^{1}\bG_2\bG_1'\mathrm{d}u,
\]
which is the Shur compliment in terms of $\int_{0}^{1}\bG_2\bG_2'\mathrm{d}u$.
By the formula for the inverse of $\int_{0}^{1} \bG\bG' \mathrm{d}u$ in terms of Shur compliment of $\int_{0}^{1}\bG_2\bG_2'\mathrm{d}u$, we can rewrite the first components in Theorem~\ref{Theorem:LSEtildebeta} as
\begin{align*}
    &\int_0^{1}\left({\left(\int_{0}^{1} \bG_{1.2}\bG_{1.2}'  \right)^{-1}} \bG_1  - \left(\int_{0}^{1} \bG_{1.2}\bG_{1.2}'  \right)^{-1} \left(\int_{0}^{1}\bG_{1}\bG_{2}'\right) \left(\int_{0}^{1}\bG_{2}\bG_{2}'\right)^{-1} \bG_2\right) \mathrm{d}\bW \\
    =& \left(\int_0^1 \bG_{1.2}\bG_{1.2}'  \right)^{-1} \int_0^1 \bG_{1.2} \mathrm{d}\bW.
\end{align*}
Similarly we can derive the formula (\ref{Formula:tildebeta2}) using the Shur compliment of $\int_{0}^{1} \bG_1 \bG_1' du$.
\end{proof}

\subsection{Proof of Theorem~\ref{Theorem:LSEhatbeta}}
\begin{proof}
Note that
\begin{align}\label{U1U2}
    \tilde{\bbeta}_i - \bbeta_i = \bB_{Ti} \bU_{Ti} = \bar{\h{\gamma}}_i \bU_{1Ti} + T^{-\frac{1}{2}} \bar{\h{\tau}}_i \bU_{2Ti}, \ i=1,2,
\end{align}
where $\bB_{Ti} = (\bar{\h{\gamma}}, T^{-\frac{1}{2}}\bar{\h{\tau}})$, and $\bU_{Ti} = (\bU_{1Ti}', \bU_{2Ti}')' = (\h{\gamma}, T^{\frac{1}{2}}\h{\tau})' \tilde\bbeta_i$.
We have the expansion
\[T(\hat{\bbeta}_{ci} - \bbeta_{ci}) = (\bI - \bbeta_{ci} \bc') T (\tilde{\bbeta}_i - \bbeta_{ci}) + O_p(T\|\tilde{\bbeta}_i - \bbeta_{ci}\|^2).\]
By (\ref{U1U2}) it follows that
\[T(\hat{\bbeta}_{ci} - \bbeta_{ci}) = (\bI - \bbeta_{ci} \bc') (\bar{\h{\gamma}}_i T \bU_{1Ti} + T^{-\frac{1}{2}} \bar{\h{\tau}}_i T \bU_{2Ti}) + O_p(T\|\tilde{\bbeta}_i - \bbeta_{ci}\|^2).\]
Since $\|\tilde\bbeta_i - \bbeta_{ci}\| = O_p(T^{-1})$ by Theorem~\ref{Theorem:LSEtildebeta}, the last term above $O_p(T\|\tilde\bbeta_i - \bbeta_{ci}\|^2) = O_p(T^{-1})$. Also Theorem~\ref{Theorem:LSEtildebeta} implies that $T\bU_{Ti}$ is convergent so the term $T^{-\frac{1}{2}} \bar{\h{\tau}}_i T \bU_{2Ti}$ vanishes. Thus, $T(\hat\bbeta_{ci} - \bbeta_{ci})$ has the same limit distribution as $(\bI - \bbeta_{ci} \bc') \bar{\h{\gamma}}_i T \bU_{1Ti}$.
\end{proof}

\subsection{Proof of Theorem~\ref{Theorem:LSEbb}}

For the proof of the Theorem~\ref{Theorem:LSEbb}, we need the following lemma. Suppose 
$$a_n \bL_n (\h{\hat{\theta}} - \h{\theta}) \to \mathcal{L},$$
where $\bL_n$ is a $d \times d$ symmetric random matrix such that $\bL_n \to \bL$ as $n \to \infty$, and $\rank(\bL_n) = r < d$. Let $\balpha = \text{Col}(\bL_n)_{\bot}$. If $\balpha' (\h{\hat{\theta}} - \h{\theta}) = o_p(1/a_n)$, we have
\begin{equation}\label{lim1}
    a_n(\h{\hat{\theta}} - \h{\theta}) \to (\bL + \balpha\balpha')^{-1} \mathcal{L}.
\end{equation}
On the other hand, note that 
\[a_n \bL_n [\balpha_{\bot}, \balpha] [\balpha_{\bot}, \balpha]' (\h{\hat{\theta}} - \h{\theta}) = a_n \bL_n \balpha_{\bot} \balpha_{\bot}' (\h{\hat{\theta}} - \h{\theta}) \to \mathcal{L},\]
which implies that
\begin{equation}\label{lim2}
     a_n(\h{\hat{\theta}} - \h{\theta}) \to \balpha_{\bot} \left(\balpha_{\bot}' \bL \balpha_{\bot} \right)^{-1} \balpha_{\bot}'\mathcal{L}.
\end{equation}

\begin{lemma}\label{inverse}
The limit distribution (\ref{lim1}) and (\ref{lim2}) are equivalent.
\end{lemma}
\begin{proof}
Let $\bL = \bU \bD \bU'$ be the eigenvalue decomposition. Then 
\[\balpha_{\bot} \left(\balpha_{\bot}' \bL \balpha_{\bot} \right)^{-1} \balpha_{\bot}'  = \bU \left(\bU' \bU \bD \bU' \bU \right)^{-1} \bU' = \bU \bD^{-1} \bU'.\]
And it can be shown that $(\bL + \balpha\balpha')^{-1} = \bU \bD^{-1} \bU' + \balpha\balpha'$. The claim follows by $\balpha' (\h{\hat{\theta}} - \h{\theta}) = o_p(1/a_n)$.
\end{proof}
\begin{proof}[Proof of Theorem \ref{Theorem:LSEbb}]
Since the projection matrix $\hat{\mathbb{P}}_{\beta_i}$ remains the same for $\hat\bbeta_{ic}$ using different normalization matrix $\bc$. We may assume that $\hat\bbeta_i$ is normalized by itself such that $\hat{\mathbb{P}}_{\beta_i} = \hat\bbeta_i \hat\bbeta'_i$. Then, by gradient conditions, we have
\begin{align}\label{bbeq}
\begin{split}
    &\frac{1}{T^2} \sum_t \left( \bM_{t-1} \bM_{t-1}' \right)
    \begin{pmatrix}
     \bI_{d_1} \otimes \bbeta_1' & \bzero \\
     \bzero & \bbeta_2' \otimes \bI_{d_2}
    \end{pmatrix}
    \begin{pmatrix}
    T\vec(\hat{\bbeta}_1\hat{\bbeta}_1' - \bbeta_1 \bbeta_1' ) \\
    T\vec(\hat{\bbeta}_2\hat{\bbeta}_2' - \bbeta_2 \bbeta_2' ) \\
    \end{pmatrix} \\
    = &\frac{1}{T} \sum_{t} \bM_{t-1} \vec(\bE_t) + o_p(1/T).
\end{split}
\end{align}
Observe that $\bM_{t-1} \bM_{t-1}' \left(\vec(\bbeta_1'), \vec(\bbeta_2) \right)' = \bzero$. On the other hand, since $\|{\bbeta}_1 {\bbeta}_1'\|_F = \|\hat{\bbeta}_1 \hat{\bbeta}_1'\|_F = r_1$, it follows that
\begin{align*}
&\left(\vec'(\bbeta_1'), \bzero \right) 
\begin{pmatrix}
     \bI_{d_1} \otimes \bbeta_1' & \bzero \\
     \bzero & \bbeta_2' \otimes \bI_{d_2} \\
\end{pmatrix}
\begin{pmatrix}
     \vec(\hat{\bbeta}_1 \hat{\bbeta}_1' - \bbeta_1 \bbeta_1') \\
     \vec(\hat{\bbeta}_2 \hat{\bbeta}_2' - \bbeta_2 \bbeta_2') \\
\end{pmatrix} \\
= &\vec'(\bbeta_1') (\bI_{d_1} \otimes \bbeta_1') \vec(\hat{\bbeta}_1 \hat{\bbeta}_1' - \bbeta_1 \bbeta_1') \\
= &\tr{\left[ \bbeta_1 \bbeta_1' \left( \hat{\bbeta}_1 \hat{\bbeta}_1' - \bbeta_1 \bbeta_1'\right)\right]} = o_p(1/T).
\end{align*}
Denote $\balpha = \left(\vec'(\bbeta_1'), \bzero'\right)'$, $\bL_n = \frac{1}{T^2}\sum_t \left( \bM_{t-1} \bM_{t-1}' \right)$. By (\ref{lim1}) and (\ref{bbeq}) we have,
\begin{align} \label{bbeq2}
\begin{split}
     &\begin{pmatrix}
     \bI_{d_1} \otimes \bbeta_1 \bbeta_1' & \bzero \\
     \bzero & \bbeta_2\bbeta_2' \otimes \bI_{d_2} \\
    \end{pmatrix}
    \begin{pmatrix}
    T\vec(\hat{\bbeta}_1'\hat{\bbeta}_1 - \bbeta_1' \bbeta_1 ) \\
    T\vec(\hat{\bbeta}_2'\hat{\bbeta}_2 - \bbeta_2' \bbeta_2 ) \\
    \end{pmatrix} \\
    = &
    \begin{pmatrix}
     \bI_{d_1} \otimes \bbeta_1 & \bzero \\
     \bzero & \bbeta_2 \otimes \bI_{d_2}
    \end{pmatrix}
    \left( \bL + \balpha \balpha' \right)^{-1} \frac{1}{T} \sum_{t} \bM_{t-1} \vec(\bE_t) + o_p(1/T),
\end{split}
\end{align}
where 
$$\left( \bL + \balpha \balpha' \right)^{-1} \frac{1}{T} \sum_{t} \bM_{t-1} \vec(\bE_t) \to \left(\int_0^1 \bG\bG' \mathrm{d}u \right)^{-1} \int_0^1 \bG (\mathrm{d}\bW)$$
converges to the right hand side of the limit distribution in Theorem~\ref{Theorem:LSEtildebeta}. And we denote
\[
\Theta := \left(\int_0^1 \bG\bG' \mathrm{d}u \right)^{-1} \int_0^1 \bG (\mathrm{d}\bW).
\]
Similarly, we can have
\begin{align} \label{bbeq3}
\begin{split}
    &\begin{pmatrix}
         \bbeta_1 \bbeta_1' \otimes \bI_{d_1} & \bzero \\
         \bzero & \bI_{d_2} \otimes \bbeta_2\bbeta_2'  \\
    \end{pmatrix}
    \begin{pmatrix}
    T\vec(\hat{\bbeta}_1\hat{\bbeta}_1' - \bbeta_1 \bbeta_1' ) \\
    T\vec(\hat{\bbeta}_2\hat{\bbeta}_2' - \bbeta_2 \bbeta_2' ) \\
    \end{pmatrix} \\
    \to &
    \begin{pmatrix}
     \bbeta_1 \otimes \bI_{d_1} & \bzero \\
     \bzero &  \bI_{d_2} \otimes \bbeta_2
    \end{pmatrix}
    \begin{pmatrix}
         \bP_{r_1,d_1} & \bzero \\
         \bzero & \bP_{r_2, d_2}
    \end{pmatrix}
    \Theta.
\end{split}
\end{align}
By adding (\ref{bbeq2}) and (\ref{bbeq3}), 
\begin{align}\label{proof10:1}
   \begin{split}
         &\bS 
    \begin{pmatrix}
    T\vec(\hat{\bbeta}_1\hat{\bbeta}_1' - \bbeta_1 \bbeta_1' ) \\
    T\vec(\hat{\bbeta}_2\hat{\bbeta}_2' - \bbeta_2 \bbeta_2' ) \\
    \end{pmatrix} \\
    \to&
    \begin{pmatrix}
      \bI_{d_1} \otimes \bbeta_1 + (\bbeta_1 \otimes \bI_{d_1})\bP_{r_1, d_1} & \bzero \\
     \bzero &  \bbeta_2 \otimes \bI_{d_2} + (\bI_{d_2}  \otimes \bbeta_2)\bP_{r_2, d_2}
    \end{pmatrix}
    \Theta
    \end{split}
\end{align}
where
\[ \bS :=
\begin{pmatrix}
     \bI_{d_1} \otimes \bbeta_1 \bbeta_1' & \bzero \\
     \bzero & \bbeta_2\bbeta_2' \otimes \bI_{d_2} \\
\end{pmatrix} +
\begin{pmatrix}
     \bbeta_1 \bbeta_1' \otimes \bI_{d_1} & \bzero \\
     \bzero & \bI_{d_2} \otimes \bbeta_2\bbeta_2'  \\
\end{pmatrix},
\]
which is not a full rank matrix, since $\left( \bI_{d_1} \otimes \bbeta_1 \bbeta_1' + \bbeta_1 \bbeta_1' \otimes \bI_{d_1}\right) \left( \bbeta_{1\bot} \bbeta_{1\bot}' \otimes \bbeta_{1\bot} \bbeta_{1\bot}' \right) = \bzero.$ However, the tricky part is that now we can use Lemma~\ref{inverse} again here, we will use the form (\ref{lim2}), where $\balpha_{\bot} = \text{Col}(\bS)$. 

First, we can verify that
$$\left( \bbeta_{1\bot} \bbeta_{1\bot}' \otimes \bbeta_{1\bot} \bbeta_{1\bot}' \right) \vec(\hat{\bbeta}_1 \hat{\bbeta}_1' - \bbeta_1 \bbeta_1') = o_p(1/T).$$
For simplicity of notations we only consider left-top block.
Denote $\bS_1:=\bI_{d_1} \otimes \bbeta_1 \bbeta_1' + \bbeta_1 \bbeta_1' \otimes \bI_{d_1}$, and eigen decomposition $\bS_1 = \bU \bD \bU'$, then
\begin{equation}\label{proof10:2}
    T \vec{(\hat{\bbeta}_1 \hat{\bbeta}_1' - \bbeta_1 \bbeta'_1)} \to \bU \bD^{-1} \bU' \left[(\bI_{d_1} \otimes \bbeta_1)L_1 + (\bbeta_1 \otimes \bI_{d_1})L_2\right],
\end{equation}
where 
$$T\vec{[(\tilde{\bbeta}_1 - \bbeta_1)' \mathbb{P}_{\beta_1\bot}] \to L_1}, \  T\vec{[\mathbb{P}_{\beta_1\bot}'(\tilde{\bbeta}_1 - \bbeta_1) ] \to L_2}.$$
We have two ways to verify \ref{proof10:2}. One is since 
$$\bS_1 \left[(\bI_{d_1} \otimes \bbeta_1)L_1 + (\bbeta_1 \otimes \bI_{d_1})L_2\right] = \left[(\bI_{d_1} \otimes \bbeta_1)L_1 + (\bbeta_1 \otimes \bI_{d_1})L_2\right],$$
then,
\begin{align*}
    \bU \bD^{-1} \bU' \left[(\bI_{d_1} \otimes \bbeta_1)L_1 + (\bbeta_1 \otimes \bI_{d_1})L_2\right] &= \bU \bD^{-1} \bU' \bS_1 \left[(\bI_{d_1} \otimes \bbeta_1)L_1 + (\bbeta_1 \otimes \bI_{d_1})L_2\right] \\
    &= \bU \begin{pmatrix}
         \bI_{dr} & \\
          & \bzero
    \end{pmatrix} \bU' \left[(\bI_{d_1} \otimes \bbeta_1)L_1 + (\bbeta_1 \otimes \bI_{d_1})L_2\right] \\
    &= \bU \begin{pmatrix}
         \bI_{dr} & \\
          & \bzero
    \end{pmatrix} \bU' \bS_1 \left[(\bI_{d_1} \otimes \bbeta_1)L_1 + (\bbeta_1 \otimes \bI_{d_1})L_2\right] \\
    &= \bU \bD \bU' \left[(\bI_{d_1} \otimes \bbeta_1)L_1 + (\bbeta_1 \otimes \bI_{d_1})L_2\right] \\
    &=\left[(\bI_{d_1} \otimes \bbeta_1)L_1 + (\bbeta_1 \otimes \bI_{d_1})L_2\right] 
\end{align*}

Another way is to note that $\bU \bD^{-1} \bU'$ is the pseudo-inverse of $\bS_1$. Denote $\bU = [\bbeta_1, \bbeta_{1\bot}] \otimes [\bbeta_1, \bbeta_{1\bot}]$.
\begin{align*}
\bU' \bS_1 \bU &= 
    [\bbeta_1, \bbeta_{1\bot}]' \otimes [\bbeta_1, \bbeta_{1\bot}]'\left[\bI_{d_1} \otimes \bbeta_1 \bbeta_1' + \bbeta_1 \bbeta_1' \otimes \bI_{d_1}\right] [\bbeta_1, \bbeta_{1\bot}] \otimes [\bbeta_1, \bbeta_{1\bot}] = \bD.
\end{align*}
Then, the limit distribution in Theorem~\ref{Theorem:LSEbb} follows from the (\ref{proof10:1}) and Lemma~\ref{inverse}.
\end{proof}

\subsection{Proof of Theorem~\ref{Theorem:MLEtheta} to \ref{Theorem:MLEbb}}

The proofs of all theorems in the Section~\ref{CMAR:TheoremsMLE} would be very similar with the proofs for the theorems in Section~\ref{CMAR:TheoremsLSE}. With the exception that $\Sigma_e$ takes the form of a Kronecker product (\ref{sigma_e}), and they are minimizing the likelihood function (\ref{like}). We only need to prove the following Lemma, and omit the rest of the proofs.
\begin{lemma}\label{MLE_rate}
For the MLE estimators,
\begin{enumerate}
    \item[(a)] $(\hat{\Sigma}_i - \Sigma_i) = O_p(\frac{1}{\sqrt{T}}), i=1,2$;
    \item[(b)] $(\hat{\Phi} - \Phi)' \bB_T = O_p(\frac{1}{\sqrt{T}})$, where $\bB_T:=(\bar{\bbeta}, T^{\frac{1}{2}} \bar{\bbeta}_{\bot})$, $\bar{\bbeta} = \bbeta (\bbeta' \bbeta)^{-1}$;
    \item[(c)] $(\hat{\bB} - \bB) = O_p(\frac{1}{\sqrt{T}})$;
    \item[(d)] $(\hat{\bD} - \bD) = O_p(\frac{1}{\sqrt{T}})$;
    \item[(e)] $(\hat{\balpha}_i - \balpha_i) = O_p(\frac{1}{\sqrt{T}}), i=1,2$;
    \item[(f)] $\hat{\bbeta}_i\hat{\bbeta}'_i - \bbeta_i \bbeta'_i = O_p(\frac{1}{T}), i=1,2$.
\end{enumerate}
\end{lemma}


\begin{proof}
The proof would be very similar to the proof of Theorem~4 in \cite{chen2021autoregressive}. Here we only address some major differences.
Without loss of generality, to simplify notations, consider the model without term $\bB$ and $\bD$,
\begin{equation}\label{Proof:lemma10modelForm2}
    \Delta \bX_t = \bA_1 \bX_{t-1} \bA_2' + \bE_t, \ t=1,\ldots,T.
\end{equation}
The proof can be extended to the model (\ref{Proof:modelForm1}) under the same idea. Denote the precision matrix $\Omega:= \Sigma^{-1} = \Omega_1 \otimes \Omega_2$, where $\Omega_i = \Sigma^{-1}_i$, $i=1,2$. And we define
\begin{align}
\begin{split}
    \mathcal{Y} &= [\Delta \vec(\bX_2), \Delta\vec(\bX_3),\ldots, \Delta\vec(\bX_{T})], \\
    \mathcal{X} &= \bU_T [\vec(\bX_1), \vec(\bX_2),\ldots, \vec(\bX_{T-1})],
\end{split}
\end{align}
where $\bU_T := (\bbeta, T^{-\frac{1}{2}} \bbeta_{\bot})' \in \mathbb{R}^{d_1 d_2 \times d_1 d_2}$ which is full rank so invertible.
$\bB_T:=(\bar{\bbeta}, T^{\frac{1}{2}} \bar{\bbeta}_{\bot})$ where $\bar{\bbeta} = \bbeta (\bbeta' \bbeta)^{-1}$. Note that $\bB_T \bU_T = \bI$. By the above definition, $\mathcal{Y}$ is made up by the stationary components, while $\mathcal{X}$ is made up by two components, one is the stationary part $\bbeta' \vec(\bX_t)$, and the other is non-stationary part $T^{-\frac{1}{2}}\bbeta_{\bot}' \vec(\bX_t)$ but controlled by the factor $T^{-\frac{1}{2}}$. Then the model~\ref{Proof:lemma10modelForm2} can be written as
\begin{equation}
    \mathcal{Y} = \Phi \bB_T \mathcal{X} + \mathcal{E},
\end{equation}
where $\Phi = \bA_2 \otimes \bA_1$, and $\mathcal{E}$ is the corresponding error matrix.
For the unrestricted CVAR model, its log like likelihood at the parameters $(\bar\Phi, \bar\Sigma)$ is
\begin{equation}
    l(\bar\Phi, \bar\Sigma) = -\frac{T-1}{2}\cdot h(\bar\Omega, \bS(\bar\Phi)),
\end{equation}
where $\bar\Omega := \bar\Sigma^{-1}$ and $\bS(\bar\Phi) := (\mathcal{Y} - \bar\Phi \bB_T\mathcal{X})(\mathcal{Y} - \bar\Phi \bB_T\mathcal{X})'/(T-1)$. Let $\check{\Phi} := \cY \cX' (\cX \cX')^{-1}$, and $\check{\bS}: =\bS(\check{\Phi}) = (\mathcal{Y} - \check\Phi \bB_T\mathcal{X})(\mathcal{Y} - \check\Phi \bB_T\mathcal{X})'/(T-1)$. Then we can show that $\bS(\Phi) \overset{\as}{\to}\Sigma$ and $\check{\bS} \overset{p}{\to} \Sigma$, and consequently $\check\Omega:=\check\bS^{-1}\overset{p}{\to} \Omega$.
Then we can prove (a) that $(\hat\Phi - \Phi)\bB_T = O_p(\frac{1}{\sqrt{T}})$ by the very similar proof of Theorem~4 in \cite{chen2021autoregressive}.
The rest of the proof of Lemma~\ref{MLE_rate} would follow by (a) and same arguments in Lemma~\ref{LSE_rate}.
\end{proof}

\end{appendices}




\end{document}